\documentclass[11pt]{article}

\usepackage[a4paper,margin=1in]{geometry}

\usepackage{siunitx}
\usepackage[T1]{fontenc}
\usepackage{verbatim}
\usepackage{amsmath,amsthm}
\usepackage{xspace}
\usepackage{pifont}
\usepackage{graphicx}
\usepackage{amssymb}
\usepackage{epic, eepic}
\usepackage{dsfont}
\usepackage{amssymb}
\usepackage{makeidx}
\usepackage{mathrsfs}
\usepackage{xcolor, colortbl}
\usepackage{subcaption}
\usepackage[numbers]{natbib}
\usepackage{enumerate}
\usepackage{enumitem}
\usepackage{afterpage}
\usepackage{mathtools}
\usepackage{epsf}
\usepackage{graphics,color}
\usepackage[export]{adjustbox}
\usepackage{tcolorbox}
\usepackage{authblk}
\usepackage[normalem]{ulem}
\RequirePackage[colorlinks,citecolor=blue,urlcolor=blue]{hyperref}
\mathtoolsset{showonlyrefs}

\usepackage{float}
\usepackage{booktabs}

\DeclareMathOperator*{\argmin}{arg\,min}

\def\eps{\varepsilon}
\def\lam{\lambda}

\def\wt{\widetilde}

\newcommand{\R}{{\mathbb R}}

\newcommand{\N}{{\mathbb N}}

\newcommand{\cmark}{\ding{51}}
\newcommand{\xmark}{\ding{55}}

\newcommand{\bd}{\begin{displaymath}}
\newcommand{\ed}{\end{displaymath}}
\newcommand{\be}{\begin{equation}}
\newcommand{\ee}{\end{equation}}
\newcommand{\bq}{\begin{eqnarray}}
\newcommand{\eq}{\end{eqnarray}}
\newcommand{\bn}{\begin{eqnarray*}}
\newcommand{\en}{\end{eqnarray*}}
\newcommand{\dl}{\delta}

\newtheorem{theorem}{Theorem}[section]

\newtheorem{convention}[theorem]{Convention}

\newtheorem{remark}[theorem]{Remark}

\newtheorem{definition}[theorem]{Definition}
\newtheorem{assumption}[theorem]{Assumption}
\numberwithin{equation}{section}

\date{\today}
\title{Nonparametric Estimation of Self- and Cross-Impact}
\author[1]{Natascha Hey\thanks{NH is supported by the Briger Family Digital Finance Lab at the Columbia Business School.}}
\author[2]{Eyal Neuman}
\author[2]{Sturmius Tuschmann\thanks{ST is supported by the EPSRC Centre for Doctoral Training in Mathematics of Random \mbox{Systems}: Analysis, Modelling and Simulation (EP/S023925/1).}}
\affil[1]{Graduate School of Business, Columbia University}
\affil[2]{Department of Mathematics, Imperial College London}

\begin{document}
\maketitle

\begin{abstract} 
We introduce an offline nonparametric estimator for concave multi-asset propagator models based on a dataset of correlated price trajectories and metaorders. Compared to parametric models, our framework avoids parameter explosion in the multi-asset case  and yields confidence bounds for the estimator. We implement the estimator using both proprietary metaorder data from Capital Fund Management (CFM) and publicly available S\&P order flow data, where we augment the former dataset using a metaorder proxy. In particular, we provide unbiased evidence that self-impact is concave and exhibits a shifted power-law decay, and show that the metaorder proxy stabilizes the calibration. Moreover, we find that introducing cross-impact provides a significant gain in explanatory power, with concave specifications outperforming linear ones, suggesting that the square-root law extends to cross-impact.  We also measure asymmetric cross-impact between assets driven by relative liquidity differences. Finally, we demonstrate that a shape-constrained projection of the nonparametric kernel not only ensures interpretability but also  slightly outperforms established parametric models in terms of predictive accuracy.
\end{abstract} 

\begin{description}
\item[Mathematics Subject Classification (2020):]  62G08, 62L05, 91G80
\item[Keywords:] market impact, cross-impact, concave price impact, nonparametric estimation, metaorder, order flow imbalance
\end{description}


\section{Introduction}\label{sec:introduction}
Price impact refers to the empirical observation that executing a large order adversely affects the price of a risky asset in a persistent manner, resulting in less favorable execution prices. It is well documented that price impact is concave with respect to trade size: larger trades tend to move prices less per unit volume than smaller trades, a property not captured by linear models. Instead, the empirical literature supports a square-root law of market impact \citep{Almgren2005,Bershova2013,Mastromatteo2014,sato2024,toth2011anomalous}, where the peak impact induced by a large \emph{metaorder} of size $Q$ is given by
\begin{equation}\label{eq:peak}
    I^{\text{peak}} = Y \sigma_D \,\text{sign}(Q)\left| \frac{Q}{V_D} \right|^{\delta},
\end{equation}
where $Y$ is a constant of order one, $V_D$ is the total daily traded volume, and $\sigma_D$ is the daily volatility. The square-root law states that empirically the exponent $\delta$ is well approximated as $0.5$. It has been shown to hold in various markets (such as equities, foreign exchange, options, and even cryptocurrencies), and is general enough to encompass different types of market liquidity, broad classes of execution strategies, and a range of trading frequencies \citep{Bershova2013,Donier_15,hey2023concave,toth2016squarerootimpactlawholds,Toth_17,Zarinelli_15}.

While the square-root law provides a simple connection between the metaorder size and its impact on the mid-price, price impact evolves dynamically, and the reaction of the mid-price to the metaorder is mostly transient. Consequently, an agent seeking to liquidate a large order must split it into smaller parts, referred to as \emph{child orders}, which are typically executed over a period of hours or days. Propagator models are a central tool for mathematically describing this decay phenomenon. They express price moves in terms of the influence of past trades and therefore capture the decay of price impact after each trade \citep{bouchaud_bonart_donier_gould_2018,Gatheral2010}. For a metaorder split into child orders $\{Q_{t_j}\}_{j=0}^{M-1}$ across $M$ time intervals, the price process $P$ follows the dynamics
\begin{equation}\label{eq:prop1D}
    P_{t_{i+1}} - P_{t_0} = \sum_{j=0}^i G_{i,j} h_c(Q_{t_j}) + \epsilon_{t_i}, \quad i = 0,\dots,M-1,
\end{equation}
where $ h_c(x) := \text{sgn}(x)|x|^c $ with $c\in (0,1]$ is a concave impact function and $(\epsilon_{t_i})_{i=0}^{M-1}$ represents the noise. Here $G=(G_{i,j})_{i,j=0}^{M-1}$ is called the propagator, which is typically decreasing in the lag $i-j$, reflecting the decaying effect of price impact. \citet{obizhaeva2013optimal} and \citet{garleanu2016dynamic}, along with follow-up papers, assume that the price impact decays exponentially over time, that is, $G_{i,j} =\kappa e^{-\rho(t_i-t_j)} $ for some positive constants $\kappa$ and $\rho$.  
On the other hand, \citet{bouchaud_bonart_donier_gould_2018} (see Chapter 13.2.1 and references therein) report on empirical observations that the propagator $G$ exhibits power-law decay in the lag, that is,  
\begin{equation}\label{power-prop} 
G_{i,j} \approx (t_i-t_j)^{-\beta} , \quad 0\leq j < i \leq M-1,
\end{equation}
where \( 0<\beta < 1 \), and these results are also supported by theoretical arguments. A well-known example à la Almgren and Chriss addresses the case where all entries of $G$ are identical, then the sum in \eqref{eq:prop1D} represents permanent price impact. If $G = \lam \mathbb{I}_M$, where $\lam>0$ and $\mathbb{I}_M\in\mathbb{R}^{M \times M}$ is the identity matrix, the sum in \eqref{eq:prop1D} represents temporary price impact (see \citep{AlmgrenChriss1,OPTEXECAC00}).

Cross-impact models provide an additional explanation for the price dynamics of a risky asset in terms of the influence of past trades of other assets in the market. Throughout this paper, we will adopt the convention introduced in Section 14.5.3 of \cite{bouchaud_bonart_donier_gould_2018} and refer to price impact as the aggregated effects of self-impact and cross-impact. When there are $d\geq 2$ assets in the market with prices denoted by $P=(P^1,\ldots,P^d)$, the evolution of the returns is given by, 
\begin{equation}\label{eq:propND}
    P^{\ell}_{t_{i+1}} - P^{\ell}_{t_0} = \sum_{k=1}^d \sum_{j=0}^i G^{(\ell,k)}_{i,j} h_{c_{(\ell,k)}}(Q^k_{t_j}) + \epsilon^{\ell}_{t_i}, \quad i = 0,\dots,M-1, \ \ell =1,\dots ,d, 
\end{equation}
where $(Q^k_{t_j})_{j=0}^{M-1}$ are the traded volumes in asset $k \in \{1,\ldots,d\}$ at each time interval, and $(\smash{G^{(\ell,k)}_{i,j}})$ are propagators that quantify the self-impact of asset $\ell$ and the cross-impact of all assets $k\neq\ell$ on the asset $\ell$. The functions $h_{c_{(\ell,k)}}$ are from the same class as in \eqref{eq:prop1D}, and can have different concavity parameters $c_{(\ell,k)}$ for self- and cross-impact (see \citep{benzaquen2017dissecting,bouchaud_bonart_donier_gould_2018, hey2024cross,coz2023cross, Mastromatteo2017,schneider2019cross,tomas2022build} for various versions of this model). \citet{bouchaud_bonart_donier_gould_2018} (see Chapter 14.5.3) report on two main empirical observations regarding cross-impact: 
\begin{itemize}
    \item[(i)] diagonal and off-diagonal elements of the propagator matrix $(\smash{G^{(\ell,k)}})$ exhibit a power-law decay in the lag, that is, $G^{(\ell,k)}(t,s) \approx (t-s)^{-\beta_{\ell k}}$, with $\beta_{\ell k} \in (0,1)$. 
\item[(ii)] most cross-correlations between price moves ($\approx 60-90\%$, depending on the timescale), are mediated by trades themselves, that is, through a cross-impact mechanism, rather than through the cross-correlation of noise terms, which are not directly related to trading.  
\end{itemize}
The controversy surrounding the shape of the propagator has raised challenging questions regarding statistical estimation methods. Several estimators for the price impact kernel in the convolution case (i.e., where $G_{i,j} = G_{i-j}$) have been proposed in \citep{benzaquen2017dissecting, Bouchaud2004, Forde:2022aa, Toth_17} and in Chapter~13.2 of \cite{bouchaud_bonart_donier_gould_2018}. These regression-based estimation methods are already common practice in the industry, but they overlook several mathematical issues, such as the ill-posedness of the least-squares estimation problem, dependencies between price trajectories, and error margins. A rigorous nonparametric estimation method for the price impact kernel in \eqref{eq:prop1D} in the linear case (i.e., $c=1$) was introduced in \cite{neuman2023statisticallearningsublinearregret}, where estimation is carried out in an online learning framework. A different approach using offline nonparametric estimation (i.e., using historical metaorder data) was developed in \cite{neuman2023offline} for the linear single-asset propagator model, and subsequently tested on data in \cite{veldman2024market}.

The main objective of this paper is to extend the nonparametric method from  \cite{neuman2023offline} to the concave multi-asset impact model in \eqref{eq:propND}, and to implement the estimator using both proprietary metaorder data by Capital Fund Management (CFM) and public order flow data.
By incorporating the empirically supported square-root law, this extension significantly improves the goodness-of-fit in kernel estimation. Moreover, adding cross-impact between assets leads to an additional improvement, which is however of smaller magnitude than the one due to concavity.

A central theoretical challenge when moving from linear to concave impact models is to study the absence of price manipulation. That is, the overall costs due to price impact should always be nonnegative for any admissible trading strategy \citep{curato2017optimal,Gatheral2010,gatheral2011exponential,gatheral.al.12,hey2023concave,huberman2004price}. For the linear multi-asset setting, sufficient conditions for convolution kernels that preclude price manipulation have been derived in both the discrete-time \cite{alfonsi2016multivariate} and the continuous-time case \cite{jaber2024optimalportfoliochoicecrossimpact}. For the concave case, already for one asset, a large class of continuous-time models with a nonlinear impact function $h$ and a nonsingular propagator $G$ admits price manipulation (see Proposition~1 of \cite{gatheral2011exponential}). Moreover, even for singular kernels such as the power-law propagator, examples of price manipulation can be constructed \cite{curato2017optimal,Gatheral2010}. To complement those findings, we extend Proposition~1 of \cite{gatheral2011exponential} to the discrete-time setting from \eqref{eq:prop1D}, demonstrating that price manipulation cannot be ruled out for nonsingular convolution kernels and nonlinear impact functions in the discrete-time case either (see Theorem~\ref{thm:price-manipulation}). As our result applies only to nonsingular kernels, we nevertheless adopt the concave model in our empirical analysis, mainly motivated by its superior goodness-of-fit.

One of the main practical challenges in applying the estimator to proprietary metaorder data lies in data scarcity. Institutional investors often execute only a single metaorder per day per asset, meaning that even long historical datasets rarely contain more than $\num{1000}$ metaorders per asset. This limited sample size is insufficient to reliably estimate high-dimensional impact kernels, especially when accounting for liquidity- or tick-size-specific effects. To address this, a proprietary dataset from CFM is enhanced using the synthetic metaorder generation procedure introduced in \cite{maitrier2025}, which is inspired by empirical findings in \cite{maitrier2025doublesquarerootlawevidence,sato2024}. These synthetic metaorders are constructed by randomly sampling from tick-by-tick data and assigning trades to proxy metaorders via a systematic prescription described in Section~\ref{sec:meta-proxy}. 

In our empirical analysis, the propagator kernel is estimated nonparametrically using two metaorder datasets: CFM's proprietary corn futures trades over a ten-year horizon and an enhanced version where we augment the dataset with synthetic metaorders. The enhanced dataset yields a smoother kernel that better captures the decay of impact, as shown in Figure~\ref{fig:cfm_enhanced_single}. In particular, we find that self-impact decays across multiple timescales without assuming any specific kernel shape. Additionally, self-impact estimations on the original CFM dataset show that concavity is a crucial driver of predictive power. Models with square-root impact systematically outperform their linear counterparts, in line with prior evidence \citep{hey2023cost}. Moreover, incorporating cross-impact across contracts with different expiries provides an additional and statistically significant gain in explanatory power. In particular, predictive accuracy improves when cross-impact is modeled as concave rather than linear, suggesting that the square-root law extends to cross-impact. Finally, we find that multi-asset estimations reveal asymmetric cross-impact patterns caused by relative liquidity differences, with more liquid assets exerting stronger and more persistent impacts on their less liquid counterparts. See Section \ref{sec:results} for details.

\begin{figure}[ht]
    \centering
    \includegraphics[width=0.6\linewidth]{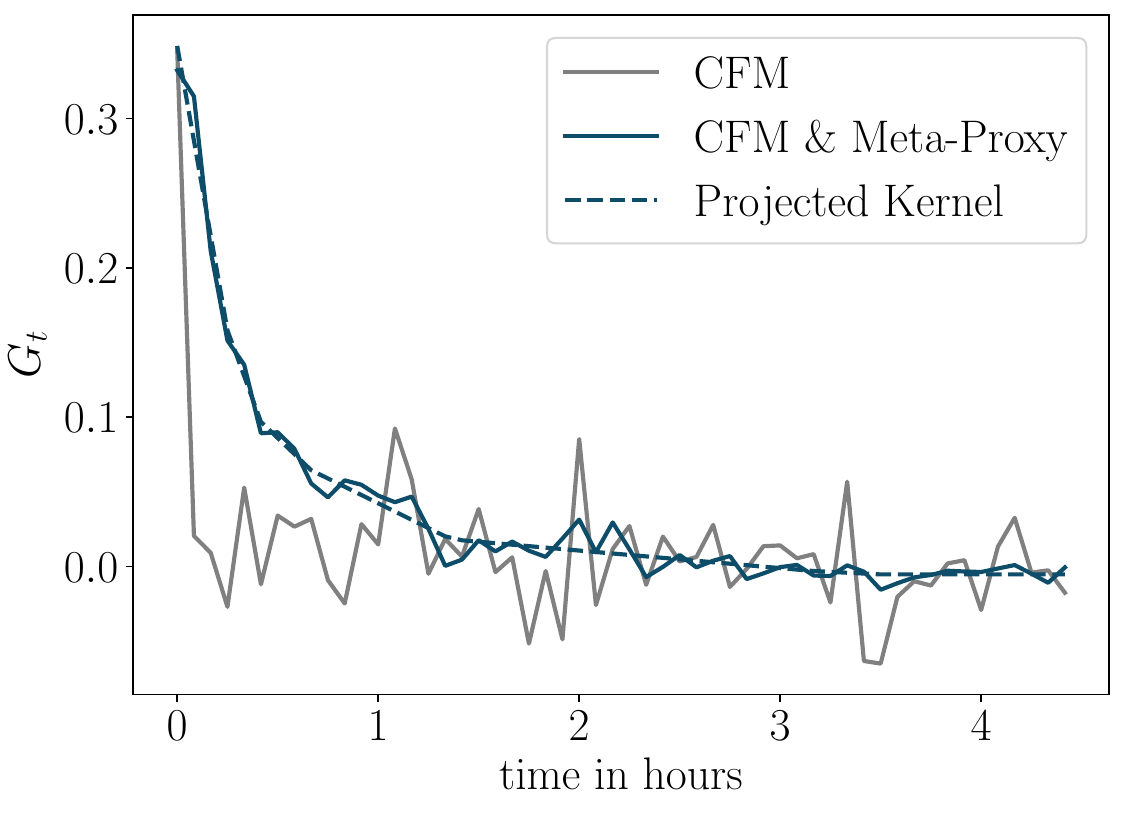}
    \caption{Estimated self-impact kernel for the square-root model ($c=0.5$) using CFM's proprietary data (gray) and the enhanced dataset with synthetic metaorders (blue) together with the projection of the kernel (dahed blue). The calibration is performed on metaorders for three corn futures with staggered maturities traded on the Chicago Mercantile Exchange (CME) between 2012 and 2022.}
    \label{fig:cfm_enhanced_single}
\end{figure}

As a complementary analysis, we apply the nonparametric estimator to a public order flow dataset for US equities, shifting the focus from the impact of a single agent to the market's price response to aggregate order flow imbalance. The nonparametric estimator confirms a shifted power-law decay of self-impact without imposing any parametric form, as illustrated in Figure~\ref{fig:of_imbalance_single}. Moreover, the projection onto a set of kernels with shape constraints markedly improves out-of-sample stability relative to the raw kernel. Predictive performance is driven primarily by concavity: square-root–type specifications outperform linear ones, with the optimal concavity for binned order flow being lower than that for metaorders, reflecting larger effective sizes in aggregated bins. Incorporating cross-impact further increases explanatory power, with concave cross-impact outperforming linear cross-impact, and we observe asymmetric cross-impact effects due to relative liquidity differences. Altogether, our results show that the setting based on order flow exhibits the same qualitative structure as the metaorder setting. See Section~\ref{sec:empirics-agg} for details.

\begin{figure}[ht]
    \centering
    \includegraphics[width=0.6\linewidth]{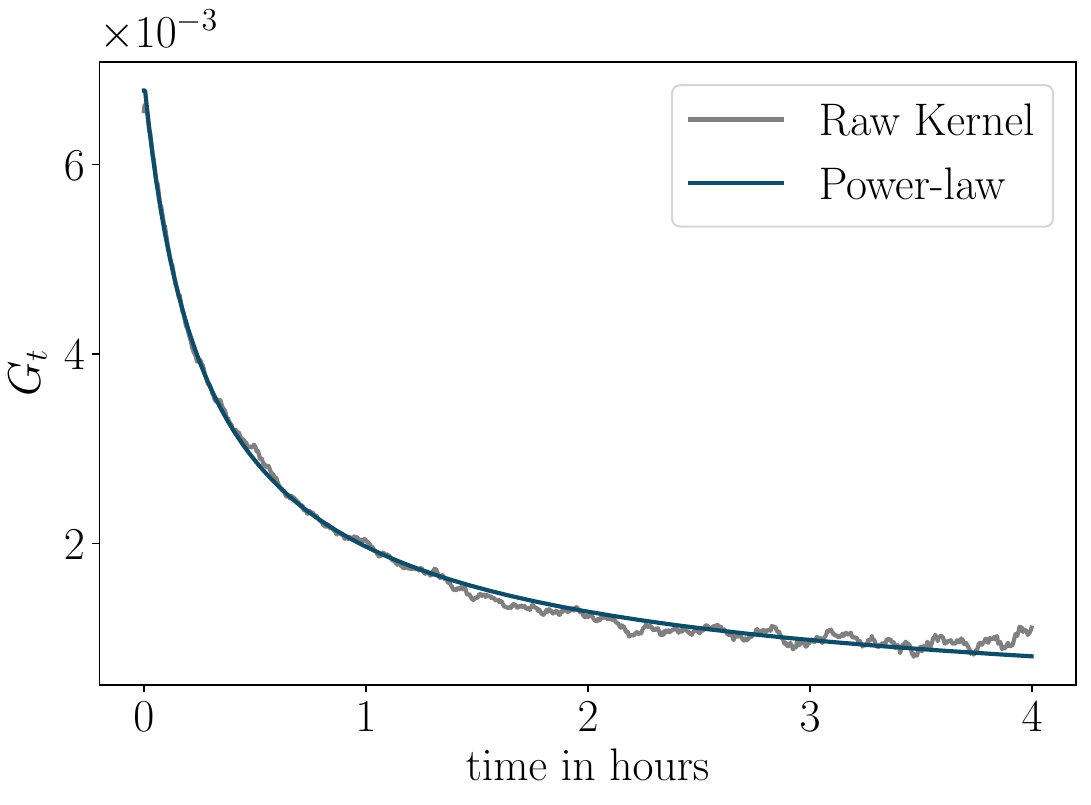}
    \caption{Estimated self-impact kernel for the square-root model ($c=0.5$) using aggregate order flow imbalance averaged over 197 stocks of the S\&P 500 in 2024 (gray) and the corresponding power-law fit (blue).}      \label{fig:of_imbalance_single}
\end{figure}

\paragraph {Our main contributions.} Below we summarize our main contributions in detail:
\begin{enumerate}
\item We extend the offline nonparametric estimator of \cite{neuman2023offline} from the linear single-asset setting to the concave multi-asset framework from \eqref{eq:propND}. Our extension allows distinct self-impact and cross-impact concavities, and yields confidence bounds for the estimator (see Theorem~\ref{thm:main}). 
Importantly, the nonparametric estimator avoids a parameter explosion in the multi-asset setting. While parametric propagators with single-exponential (1-EXP), double-exponential (2-EXP), or power-law (POWER) decay require up to $\mathcal{O}(d^2)$ tuned parameters and thus a grid search that scales quadratically in the number of assets $d$, the nonparametric propagator (RAW) and its projection enforcing shape constraints (PROJ) require none (see Tables~\ref{tab:model-parameter-comparison-d-assets} and \ref{tab:model-parameter-comparison}).

\item We prove that the discrete-time impact model in \eqref{eq:prop1D} admits price manipulation whenever $G$ is a convolution kernel that is continuous in $0$ with $G(0)>0$ and $h:\R\to\R$ is continuous, odd, and not linear (see Theorem \ref{thm:price-manipulation}). This result extends Proposition~1 of \cite{gatheral2011exponential} to the discrete-time setting.

\item We extend the metaorder proxy method from \cite{maitrier2025}, originally employed to augment sparse proprietary datasets, to the multi-asset setting. The proxy preserves key microstructure features, such as sign persistence and realistic size distribution, and materially stabilizes kernel estimates prior to projection (see Figure~\ref{fig:cfm_enhanced_single} and Section~\ref{sec:meta-proxy}).  As a result, we show on CFM's proprietary metaorder data that the proxy improves the predictive power of single- and multi-asset impact models (see Table~\ref{tab:R2}).

\item On CFM's proprietary metaorder dataset, we confirm that self-impact is concave and decays across multiple timescales, with square-root impact doubling the prediction performance in terms of $R^2$ compared to the linear case (see Table~\ref{tab:metaorder-model-comparison} and Figure~\ref{fig:R2_concave}). Moreover, we find that the proxy enhancement smoothens the decay and improves the $R^2$ from 4.6\% to 4.8\%, and that adding additional assets raises it further to $6.3-6.5\%$ (see Figures~\ref{fig:cfm_enhanced_single},~\ref{fig:Kernel1dloglog} and Table~\ref{tab:R2}), suggesting that the square-root law holds for cross-impact as well (see Table~\ref{tab:metaorder-crossimpact-comparison} and Figure~\ref{fig:R2_concave}). We also measure asymmetric cross-impact dynamics caused by relative liquidity differences (see Figure \ref{fig:MultiAsset}).

\item On the public order flow dataset, we confirm that self-impact is concave and that the estimated decay kernel is best fitted by a shifted power-law (see Figures~\ref{fig:of_imbalance_single}, \ref{fig:kernels} and Table~\ref{tab:concavity-comparison}), with the projected kernel significantly outperforming parametric models (see Table~\ref{tab:model-comparison}). The introduction of cross-impact further improves prediction performance, with concave specifications outperforming linear ones as in the metaorder case (see Table~\ref{tab:model-comparison-cross}). Finally, asymmetric cross-impact can be observed for highly correlated stocks as well (see Figure~\ref{fig:MultiAsset-stocks}).
\end{enumerate}

\begin{table}[htb]
  \centering
  \begin{tabular}{||c|c|c|c||}
    \hline
     Model      & Propagator              & Tuned parameters  & Search dimension \\ 
    \hline\hline
    1-EXP       & $e^{-\rho t}$                & $\rho$                   &  $d^2$            \\
    2-EXP       & $w_1e^{-\rho_1t} +(1-w_1)e^{-\rho_2t}$ & $\rho_1,\rho_2$    &  $2d^2$           \\
    POWER       & $(t+\tau)^{-\beta}$          & $\beta,\tau$             & $2d^2$             \\
    RAW    & $\wt{G}(t)$                 & None                       & 0                \\
    PROJ        & $G(t)$   & None                        & 0                \\
    \hline
  \end{tabular}
  \caption{Propagators, tuned parameters per propagator, and number of tuned parameters needed for $d$ assets in the parametric models 1-EXP, 2-EXP, POWER and the nonparametric models RAW, PROJ described above. The number of concavity parameters is the same for all propagators, and therefore omitted for the model comparison.}
  \label{tab:model-parameter-comparison-d-assets}
\end{table}

\paragraph{Structure of the paper.} The remainder of this paper is structured as follows. Section~\ref{sec:estimation} presents the theoretical foundations, that is, the definition of the concave multivariate propagator model and the offline estimator used for kernel calibration in Section~\ref{subsec:dataset}, a price manipulation result for the model in Section~\ref{subsec:manipulation}, and a confidence interval for the estimator in Section \ref{subsec:generalestimation}. Section~\ref{sec:empirics} contains the empirical analysis, including the construction of synthetic metaorders in Section~\ref{sec:meta-proxy}, the methodology needed for reliable kernel estimation in Section~\ref{sec:methods}, the empirical findings for proprietary metaorder data in Section~\ref{sec:results} and the empirical findings for public order flow data in Section~\ref{sec:empirics-agg}.

\section{Nonparametric Impact Estimation from Offline Data}\label{sec:estimation}
This section introduces the offline dataset and methodology used to estimate the true impact kernel denoted by $ \boldsymbol{G}^* =\big((G^*)^{(k,\ell)}\big)$, that was introduced in \eqref{eq:propND}. The static dataset consists of either metaorders or order flow imbalances and their corresponding price trajectories observed across multiple time periods and assets. 
A least-squares framework is used to estimate $\boldsymbol{G}^*$, in both an unconstrained and a constrained setting, where the constrained estimator enforces admissibility conditions derived from no-arbitrage arguments in the linear propagator model.

\subsection{Offline Dataset}\label{subsec:dataset}
Throughout, we fix the number of assets $d\in\N$, a time horizon $T>0$, and an equidistant partition $\mathbb{T} := \{0 = t_0,t_1,\ldots , t_M = T\}$ of the interval $[0,T]$ consisting of $M\in\N$ subintervals. The offline dataset is defined as follows.
\begin{definition} \label{def:dataset}
Let $N \in \mathbb{N}$ be the number of episodes in the dataset and consider
\begin{equation}
    \mathcal{D}  = \left\{ (\boldsymbol{P}_{t_i}^{(n)})^{M}_{i=0},(\boldsymbol{Q}_{t_i}^{(n)})^{M-1}_{i=0}\,\Big|\, n = 1,\dots,N\right\},
\end{equation}
where $\boldsymbol{P}^{(n)}_{t_i} = \left( (P_{t_i}^1)^{(n)}, \ldots, (P_{t_i}^d)^{(n)} \right)^\top\in\R^d$ and $\boldsymbol{Q}^{(n)}_{t_i} = \left( (Q_{t_i}^1)^{(n)}, \dots, (Q_{t_i}^d)^{(n)}\right)^\top\in\R^d$ respectively capture the observed price and traded volume for all $d$ assets at time $t_i$ in the $n$-th episode. For $n=1,\ldots, N$, 
define $\smash{\boldsymbol{P}^{(n)}:= (\boldsymbol{P}^{(n)}_{t_i})_{i=0}^{M}}$ and $\smash{\boldsymbol{Q}^{(n)}:= (\boldsymbol{Q}^{(n)}_{t_i})_{i=0}^{M-1}}$.

We call $\mathcal{D}$ an offline dataset if $(\boldsymbol{P}^{(n)},\boldsymbol{Q}^{(n)})_{n=1}^N$ 
are realizations of random variables defined on a probability space $(\Omega,\mathcal{F},\mathbb{P})$
satisfying the following properties: for each episode $n\in\{1,\ldots, N\}$, $\boldsymbol{Q}^{(n)}$ is measurable with respect to the $\sigma$-algebra $\mathcal{F}_{n-1}$, where 
\begin{equation}
    \mathcal{F}_{n-1} = \sigma\left\{\big(\boldsymbol{P}^{(r)}\big)_{r=1}^{n-1},\big(\boldsymbol{Q}^{(r)}\big)_{r=1}^{n-1} \right\},
\end{equation}
and there exist $\mathcal{F}_n$-measurable random variables $\boldsymbol{\epsilon}^{(n)}_{t_i} = \big((\epsilon^1_{t_i})^{(n)},\ldots,(\epsilon^d_{t_i})^{(n)}\big)^\top$ 
such that the price evolution follows a concave multi-asset propagator model with additive noise,

\begin{equation}\label{eq:returns}
    \boldsymbol{P}^{(n)}_{t_{i+1}} -  \boldsymbol{P}^{(n)}_{t_0} = \sum_{j=0}^{i} \boldsymbol{G}^*_{i-j}h(\boldsymbol{Q}^{(n)}_{t_j}) + \boldsymbol{\epsilon}^{(n)}_{t_i},\quad i = 0,\ldots,M-1,
\end{equation}
where $ \boldsymbol{G}^* \in \mathbb{R}^{M\times  d\times d } $ is the true (unknown) convolution kernel and $h:\R\to\R$ is a continuous, increasing function that is concave on $[0,\infty)$ and is applied elementwise. The random variables satisfy $\smash{\mathbb{E}[\boldsymbol{\epsilon}_{t_i}^{(n)}| \mathcal{F}_{n-1}] = 0}$ and $\smash{\mathbb{E}[(\boldsymbol{\epsilon}^{(n)}_{t_i})^\top\boldsymbol{\epsilon}^{(n)}_{t_i}] < \infty}$ for all $i = 0,\dots,M-1$.
\end{definition}
\begin{remark}\label{rem:h}
A common specification of the concave impact function $h$ in Definition~\ref{def:dataset} is
\be\label{eq:hc}
h_c(x)=\operatorname{sgn}(x)|x|^c,\quad c\in(0,1],
\ee
which we also use in the empirical analysis in Section~\ref{sec:empirics}. In practice, \eqref{eq:returns} can be generalized by letting the concavity depend on the asset pair. A case of particular interest is a two-parameter version in which self- and cross-impact are governed by $h_{c_S}$ and $h_{c_X}$ for constants $c_S,c_X\in(0,1]$, respectively. While Definition~\ref{def:dataset} does not cover this variant, it can be incorporated with minor notational changes, and all subsequent arguments carry over unchanged.
\end{remark}
Define $\smash{\boldsymbol{\epsilon}^{(n)}:=(\boldsymbol{\epsilon}^{(n)}_{t_i})_{i=0}^{M-1}}$. For any $k\in\N$ and $u,v\in\R^{k}$, we denote by $\langle u,v\rangle_{\R^k}$ the inner product in $\R^k$ and by $\|v\|_{\R^k}$ the Euclidean norm. In line with \cite{neuman2023offline}, the noise is assumed to be conditionally sub-Gaussian. This assumption enables the application of high-probability bounds to the estimated propagator matrices. Moreover, this setting is versatile enough to allow for dependence between the price trajectories in consecutive trading periods.
\begin{assumption}\label{assumption}
Given an offline dataset $\mathcal{D}$ of size $N\in \mathbb{N}$, there exists a known constant $R>0$ such that for all $n = 1,\ldots,N$,
\begin{equation}
    \mathbb{E}\Big[\exp\big(\langle v,\boldsymbol{\epsilon}^{(n)}\rangle_{\R^{M d}}\big) | \mathcal{F}_{n-1}\Big] \leq \exp\left(\frac{R^2\|v\|_{\R^{M d}}^2}{2}\right),\quad  \text{for all } v \in \mathbb{R}^{M d},
\end{equation}
that is, each $\boldsymbol{\epsilon}^{(n)}$ is $R$-conditionally sub-Gaussian with respect to $\mathcal{F}_{n-1}$.
\end{assumption}
\subsection{Price Manipulation and Admissible Kernels}\label{subsec:manipulation}
Let $\boldsymbol{G} \in \mathbb{R}^{M\times  d\times d }$ be a propagator kernel and $\boldsymbol{Q}:=(\boldsymbol{Q}_{t_i})_{i=0}^{M-1}$ with 
$$\boldsymbol{Q}_{t_i} = \left( (Q_{t_i}^1), \ldots, (Q_{t_i}^d) \right)^\top\in\R^d$$
be a sequence of trades in the $d$ assets. In line with Lemma~1 of \cite{alfonsi2016multivariate}, it follows from \eqref{eq:returns} that the total execution costs corresponding to $\boldsymbol{Q}$ caused by the self- and cross-impact associated with $\boldsymbol{G}$ are given by
\be\label{eq:costs}
\mathcal{C}_{\mathbb{T}}(\boldsymbol{Q}):=\frac{1}{2}\sum_{i,j=0}^{M-1}\boldsymbol{Q}_{t_i}^\top \bar{\boldsymbol{G}}_{i-j}h(\boldsymbol{Q}_{t_j}),
\ee
where $\bar{\boldsymbol{G}} \in \mathbb{R}^{(2M-1)\times  d\times d }$ is given by
$$
\bar{\boldsymbol{G}}_i=
\begin{cases}
    \boldsymbol{G}_{i},\phantom{\frac{1}{2}(\boldsymbol{G}_0+\boldsymbol{G}_0^\top)\boldsymbol{G}_{-i}^\top} 1\leq i\leq M-1,\\
    \frac{1}{2}(\boldsymbol{G}_0+\boldsymbol{G}_0^\top),\phantom{\boldsymbol{G}_{i}\boldsymbol{G}_{-i}^\top} i=0,\\
    \boldsymbol{G}_{-i}^\top,\phantom{\boldsymbol{G}_{i}\frac{1}{2}(\boldsymbol{G}_0+\boldsymbol{G}_0^\top)} 1\leq -i\leq M-1.
\end{cases}
$$
The condition 
\be\label{eq:absence}
\mathcal{C}_{\mathbb{T}}(\boldsymbol{Q})\geq 0,\quad\text{for all }\boldsymbol{Q}=(\boldsymbol{Q}_{t_i})_{i=0}^{M-1}\text{ and all }\mathbb{T},
\ee
is a regularity condition for the underlying impact model that rules out the possibility of profiting from exploiting one's own market impact (see, e.g., \cite{jaber2024optimalportfoliochoicecrossimpact,alfonsi2016multivariate,Gatheral2010,gatheral2011exponential,gatheral.al.12,hey2024cross}). For linear cross-impact ($h(x)=x$), conditions that prevent price manipulation have been derived in Theorem~2 of \cite{alfonsi2016multivariate} for the discrete-time case and in Theorem~2.14 of \cite{jaber2024optimalportfoliochoicecrossimpact} for the continuous-time case (see also Theorem~1.2 of \cite{neuman2025mercer} for their equivalence). To preclude price manipulation, kernels must be nonnegative, convex, symmetric, and nonincreasing. However, the present work incorporates concave impact functions $h$, for which the theoretical understanding of admissibility is more involved. For the continuous-time case, it is known that already for a single asset any model with a nonlinear impact function $h$ and a kernel $\boldsymbol{G}$ that is nonsingular at time zero admits price manipulation (see Proposition~1 of \cite{gatheral2011exponential}). Moreover, even for singular kernels such as the power-law propagator, examples for price manipulation can be constructed (see \cite{curato2017optimal,Gatheral2010}). The following theorem shows that there is not much hope to rule out price manipulation for the discrete-time case either. 
\begin{theorem}\label{thm:price-manipulation}
Let $H:[0,T]^2\to\R$ and suppose that there exists $t^*\in[0,T)$ such that $H$ is continuous in $(t^*,t^*)$ with $H(t^*,t^*)>0$. Let $h:\mathbb{R}\to\mathbb{R}$ be continuous, odd function which is not linear, that is, not of the form $h(x)=qx$ for some $q\in\R$.
Then for every $M\ge3$ there exist $t_0<\ldots<t_{M-1}$ and $x_0,\dots,x_{M-1}\in\mathbb{R}$ such that
\[
\sum_{i,j=0}^{M-1} x_i\,H(t_i,t_j)\,h(x_j)\;<\;0.
\]
In particular, given a kernel $G:[0,T]\to\R$ that is continuous in $0$ with $G(0)>0$, it holds that
\be
\sum_{i,j=0}^{M-1} x_i\,G(|t_i-t_j|)\,h(x_j)\;<\;0.
\ee
\end{theorem}
The proof of Theorem~\ref{thm:price-manipulation} is given in the Appendix~\ref{appendix:proof}.
\begin{remark}
By \eqref{eq:costs} and \eqref{eq:absence}, for any $G$ and $h$ as in Theorem~\ref{thm:price-manipulation}, the corresponding single-asset impact model admits price manipulation. In particular, Theorem~\ref{thm:price-manipulation} extends Proposition~1 of \cite{gatheral2011exponential} to the discrete-time setting.
\end{remark}

\begin{remark}
A concave impact model in which the absence of price manipulation can be ensured is the Alfonsi-Fruth-Schied model (see \cite{alfonsi2014optimal,alfonsi2010optimal2,alfonsi2010optimal,gatheral2011exponential,predoiu2011optimal}). For a related multi-asset framework, necessary conditions for this absence in terms of kernel matrix properties were recently established in \cite{hey2024cross} (see Lemma 5.1 therein). Notably, when the concavity parameter $c$ of the impact function $h_c$ from \eqref{eq:hc} is contained in $[0.5,1]$, the conditions closely resemble those in the linear setting. While the Alfonsi-Fruth-Schied model provides a useful framework to study concave impact, it is not suitable for our purposes, since it relies on the assumption that the impact decays exponentially. Moreover, its calibration requires additional data compared to the model in \eqref{eq:returns}, which is often not available in proprietary metaorder datasets. 
\end{remark}

In this work, given that Theorem~\ref{thm:price-manipulation} only holds for nonsingular kernels and that the introduction of shape constraints in \cite{neuman2023offline} improves the kernel estimation, we adopt and slightly relax the admissibility conditions from the linear case (see Theorem~2.14 in \cite{jaber2024optimalportfoliochoicecrossimpact} and Theorem~2 in \cite{alfonsi2016multivariate}). As we will see in our empirical analysis, this will significantly improve the model's out-of-sample performance (see Section \ref{sec:empirics}). To wit, we no longer enforce symmetry of the propagator matrix to allow for asymmetric cross-impact across assets. The resulting class of admissible kernels includes all $ \boldsymbol{G} \in \mathbb{R}^{M\times  d\times d } $ such that the induced matrix-valued map $\{0,\ldots, M-1\} \rightarrow \mathbb{R}^{d \times d},\ i\mapsto \boldsymbol{G}_i$ is nonnegative, nonincreasing, and convex. Formally, the set of admissible kernels is defined as follows:
\begin{equation}\label{eq:admissible}
\mathcal{G}_{\mathrm{ad}} := \left\{ 
   \boldsymbol{G} = \left(\boldsymbol{G}_i\right)^{M-1}_{i=0} \ \left|\
   \begin{aligned} 
   & v^\top \boldsymbol{G}_{i} v\geq 0,&& 0\leq i\leq M-1,\\
   & v^\top \boldsymbol{G}_{i+1} v \leq v^\top \boldsymbol{G}_i v,&& 0\leq i\leq M-2, \\
   & v^\top (\boldsymbol{G}_{i+2}-\boldsymbol{G}_{i+1}) v\leq v^\top (\boldsymbol{G}_{i+1}-\boldsymbol{G}_i) v,\ && 0\leq i\leq M-3,\\
   &\text{for all }  v \in \mathbb{R}^d.
   \end{aligned} \right.
\right\}.
\end{equation}
In \eqref{eq:admissible}, the first, second, and third conditions respectively ensure nonnegativity, monotonicity, and convexity of $\boldsymbol{G}$, respectively. This translates into requiring that the self-impact terms are nonnegative and dominate the cross-impact terms, and that all impact terms decay in a convex manner.
Typical examples of cross-impact propagator matrices in $\mathcal{G}_{\mathrm{ad}}$ include kernels with exponential or power-law decay; see Example 2.2 in \cite{jaber2024optimalportfoliochoicecrossimpact} for a comprehensive overview.
\begin{remark}\label{rem:Gad}
In \eqref{eq:admissible}, we may additionally require the following symmetry condition,
$$
\boldsymbol{G}_{i}=\boldsymbol{G}_{i}^\top,\quad 0\leq i\leq M-1,
$$
without affecting any of the results below. As expounded above, this ensures the absence of price manipulation in the case of linear impact (see Theorem~2 in \cite{alfonsi2016multivariate}), but it restricts attention to symmetric cross-impact. However, because our focus is on concave impact, where such no-manipulation guarantees generally fail, we retain the baseline definition of $\mathcal{G}_{\mathrm{ad}}$ in \eqref{eq:admissible}.
\end{remark}


\subsection{Nonparametric Estimation}\label{subsec:generalestimation}
The true unknown propagator matrix $\boldsymbol{G}^*\in\R^{M\times d\times d}$ is estimated via a least-squares method with constraints. For this, we introduce the following convention. 
\begin{convention}
Given $\boldsymbol{G}\in\R^{M\times d\times d}$, we identify $\R^{M\times d\times d}$ with $\R^{Md^2}$ via the vectorization that stacks entries so that $\ell\in\{1,\ldots,d\}$ is the slowest-varying index, then $i\in\{0,\ldots,M-1\}$, and $k\in\{1,\ldots,d\}$ is the fastest:
\be\label{eq:Gv}
\boldsymbol G=
\left( 
G^{(1,1)}_{0}, \ldots ,G^{(1,d)}_{0} ,\ldots ,G^{(1,1)}_{M-1} , \ldots ,G^{(1,d)}_{M-1},\ldots,
G^{(d,1)}_{0}, \ldots ,G^{(d,d)}_{0} ,\ldots ,G^{(d,1)}_{M-1} , \ldots ,G^{(d,d)}_{M-1}\right)^\top. 
\ee
By this identification, we will always use the same notation for the tensor and its vectorization, writing $\boldsymbol G$ for either, with the intended meaning clear from context.
\end{convention}
For each episode $n\in\{1,\ldots, N\}$, define the vector of observed returns $\boldsymbol{y}^{(n)} \in \mathbb{R}^{M d}$ as 
\begin{equation}\label{eq:yn}
\boldsymbol{y}^{(n)} = \begin{pmatrix}
    \big((P_{t_{i+1}}^{1})^{(n)}- (P_{t_{0}}^{1})^{(n)}\big)_{i=0}^{M-1}\\
    \vdots\\
    \big((P_{t_{i+1}}^{d})^{(n)}- (P_{t_{0}}^d)^{(n)}\big)_{i=0}^{M-1} \end{pmatrix}.
\end{equation}
Moreover, define the auxiliary matrices $\boldsymbol{D}^{(n)} \in \mathbb{R}^{M \times (M  d)}$ and $\boldsymbol{U}^{(n)} \in \mathbb{R}^{(M d) \times (M d^2)}$ capturing the transformed traded volume as
$$
\boldsymbol{D}^{(n)} = 
\begin{pmatrix}
h\big((Q^1_{t_0})^{(n)}\big) &\dots & h\big((Q^d_{t_0})^{(n)}\big) & 0 & \dots&0 \\
h\big((Q^1_{t_1})^{(n)}\big) & \dots & h\big((Q^d_{t_1})^{(n)}\big) & h\big((Q^1_{t_0})^{(n)}\big) &\dots & 0 \\
\vdots & \vdots & \vdots & \vdots & \ddots & \vdots \\
h\big((Q^1_{t_{M-1}})^{(n)}\big) &\dots & h\big((Q^d_{t_{M-1}})^{(n)}\big)  &h\big((Q^1_{t_{M-2}})^{(n)}\big)& \dots & h\big((Q^d_{t_0})^{(n)}\big) 
\end{pmatrix}, 
$$
and 
\be\label{eq:Un}
\boldsymbol{U}^{(n)} = 
\begin{pmatrix}
\boldsymbol{D}^{(n)} & 0 & \dots &0 \\
0 & \boldsymbol{D}^{(n)} & \ddots & \vdots \\
\vdots & \ddots &\ddots & 0 \\
0 & \dots & 0 & \boldsymbol{D}^{(n)}
\end{pmatrix}.
\ee
Then, by \eqref{eq:returns}, \eqref{eq:Gv}, and \eqref{eq:Un}, the observed returns in \eqref{eq:yn} satisfy the equation
\begin{equation}\label{eq:returns-equation}
    \boldsymbol{y}^{(n)} = \boldsymbol{U}^{(n)} \boldsymbol{G}^* + \boldsymbol{\epsilon}^{(n)}.
\end{equation}
In line with \cite{neuman2023offline}, this motivates us to consider the following constrained least-squares estimator,
\begin{equation}\label{eq:constrained}
    \boldsymbol{G}_{N,\lambda} := \argmin_{\boldsymbol{G} \in \mathcal{G}_{\mathrm{ad}}} \left( \sum_{n=1}^N \left\| \boldsymbol{y}^{(n)} - \boldsymbol{U}^{(n)} \boldsymbol{G} \right\|_{\R^{M d}}^2 + \lambda \left\| \boldsymbol{G} \right\|_{\R^{M d^2}}^2 \right),
\end{equation}
where $\lambda >0$ is a regularization parameter, and $\mathcal{G}_{\mathrm{ad}}$ is the set of admissible kernels from \eqref{eq:admissible}.
\begin{remark}
As noted in \cite{neuman2023offline}, the regularization parameter $\lambda >0$ ensures the strong convexity of \eqref{eq:constrained} in $\boldsymbol{G}$. Together with the fact that $\mathcal{G}_{\mathrm{ad}}$ in \eqref{eq:admissible} is nonempty, closed, and convex, we obtain that $\boldsymbol{G}_{N,\lambda}$ is uniquely defined. This analogously holds when $\mathcal{G}_{\mathrm{ad}}$ is modified as in Remark~\ref{rem:Gad}. This regularization parameter is necessary in general, as it ensures the invertibility of the matrix on the left-hand side of \eqref{eq:unconstrained-solution}. In the multivariate setting, asset-specific regularization parameters $\Lambda=(\lambda^{(\ell,k)})_{\ell,k=1}^d$ can be introduced in practice. The optimization problem \eqref{eq:constrained} then becomes 
\begin{equation}\boldsymbol{G}_{N,\Lambda} := \argmin_{\boldsymbol{G} \in \mathcal{G}_{\mathrm{ad}}}  \left( \sum_{n=1}^N \left\| \boldsymbol{y}^{(n)} - \boldsymbol{U}^{(n)} \boldsymbol{G} \right\|_{\R^{M d}}^2  \right.\left. + \sum_{\ell,k=1}^d \lambda^{(\ell,k)} \left\| (G_{i}^{(\ell,k)})_{i=0}^{M-1} \right\|_{\R^{M}}^2 \right).
\end{equation} 
\end{remark}
The constrained estimator $\boldsymbol{G}_{N,\lambda}$ in \eqref{eq:constrained} can be computed by projecting the corresponding unconstrained least-squares estimator onto the set $\mathcal{G}_{\mathrm{ad}}$. Indeed, letting, \begin{equation}\label{eq:unconstrained}
\tilde{\boldsymbol{G}}_{N,\lambda}:=\argmin_{\boldsymbol{G} \in \R^{Md^2}} \left( \sum_{n=1}^N \left\| \boldsymbol{y}^{(n)} - \boldsymbol{U}^{(n)} \boldsymbol{G} \right\|_{\R^{M d}}^2 + \lambda \left\| \boldsymbol{G} \right\|_{\R^{M d^2}}^2 \right),
\end{equation}
we have that
\begin{equation}
\begin{aligned}\label{eq:constrained-rewritten}
    \boldsymbol{G}_{N,\lambda} 
    &= \argmin_{\boldsymbol{G} \in \mathcal{G}_{\mathrm{ad}}} \left\| W^{1/2}_{N,\lambda} \left( \boldsymbol{G} - \tilde{\boldsymbol{G}}_{N,\lambda} \right) \right\|_{\R^{M d^2}}^2,
\end{aligned}
\end{equation}
with
\begin{equation}\label{eq:unconstrained-solution}
    W_{N,\lambda} := \sum_{n=1}^N (\boldsymbol{U}^{(n)})^\top \boldsymbol{U}^{(n)} + \lambda \mathbb{I}_{M  d^2}, \quad
    \tilde{\boldsymbol{G}}_{N,\lambda} := W^{-1}_{N,\lambda} \sum_{n=1}^N (\boldsymbol{U}^{(n)})^\top \boldsymbol{y}^{(n)},
\end{equation}
where $\boldsymbol{y}^{(n)}$ and $\boldsymbol{U}^{(n)}$ are defined in \eqref{eq:yn} and \eqref{eq:Un}, respectively, and $\mathbb{I}_{M  d^2}$ denotes the identity matrix of size $(Md^2)\times (Md^2)$. That is, the unconstrained problem \eqref{eq:unconstrained} is solved explicitly by \eqref{eq:unconstrained-solution}, and the constrained problem \eqref{eq:constrained} can be solved efficiently by means of \eqref{eq:constrained-rewritten} and  standard convex optimization packages (see, e.g., \cite{diamond2016cvxpy}).

This estimation procedure extends the functional regression with structural constraints on the kernel shape from \cite{neuman2023offline} to the multivariate setting. In particular, the confidence region of the estimated impact coefficients $\boldsymbol{G}_{N,\lambda}$ in terms of the observed data can be derived analogously to Theorem 2.14 therein. 
\begin{theorem}\label{thm:main}
Suppose that the true propagator matrix $\boldsymbol{G}^*$ is contained in the set of admissible kernels $\mathcal{G}_{\mathrm{ad}}$ from \eqref{eq:admissible} and that Assumption \ref{assumption} holds. For any $\lambda > 0$, let the constrained least-squares estimator $\boldsymbol{G}_{N,\lambda}$ be defined as in \eqref{eq:constrained}. Then, for all $\lambda > 0$ and $\delta \in (0,1)$, with probability at least $1 - \delta$,
\begin{equation}
\left\|{W_{N,\lambda}^{1/2} }\big( \boldsymbol{G}_{N,\lambda} - \boldsymbol{G}^* \big) \right\|_{\R^{M d^2}}
\leq R \left( 2 \log \left( \frac{\det(W_{N,\lambda})}{\delta^2\lambda^{M d^2}}  \right) \right)^{1/2} 
+ \lambda \left\| {W^{-1/2}_{N,\lambda}} \boldsymbol{G}^* \right\|_{\R^{M d^2}},
\end{equation}
where $W_{N,\lambda} = \sum_{n=1}^N (\boldsymbol{U}^{(n)})^\top \boldsymbol{U}^{(n)} + \lambda \mathbb{I}_{M  d^2}$ and $R>0$ is the constant from Assumption~\ref{assumption}. 
\end{theorem}
\begin{remark}
Theorem~\ref{thm:main} extends Theorem~2.14 from \cite{neuman2023offline} into two important directions. First, it allows the introduction of a concave impact function $h$ as in \eqref{eq:returns}. More importantly, it allows the analysis of a portfolio of $d$ assets with cross-impact across them as specified in \eqref{eq:returns}. Note that the convergence rate in Theorem~\ref{thm:main} remains similar to the one in Theorem 2.14 from \cite{neuman2023offline}, with the only difference being the dimensional scaling. Since the estimator now involves $d^2$ kernels over $M$ time steps, the determinant term contributes a factor of $\lambda^{-M d^2}$ instead of $\lambda^{-M}$.
\end{remark}
The proof of Theorem~\ref{thm:main} is deferred to the Appendix~\ref{appendix:proof}.

\section{Empirical Analysis}\label{sec:empirics}

This section presents the nonparametric estimation of self- and cross-impact propagators for metaorders on futures contracts and for aggregate order flow on S\&P 500 stocks. In both cases, self-impact models are calibrated and evaluated first, followed by cross-impact models. Moreover, their prediction performance is compared to established parametric models with single-exponential, double-exponential, and power-law decay. As the number of observed metaorders in a single contract is limited, a proxy mechanism is employed to synthetically augment the dataset. Using corn futures as an example, we show that this augmentation enables more robust estimation of the impact kernel. The results indicate that concave impact functions outperform their linear counterparts in the single-asset case. Accordingly, the identified concavity parameter is retained for subsequent model calibration within the multivariate framework involving contracts with different maturities.

\subsection{Datasets}\label{sec:dataset}
Our study draws on a combination of proprietary and public datasets including various asset classes. These datasets serve distinct but complementary research purposes: the first two focus on the estimation of price impact from metaorders in futures markets, while the third examines order flow imbalance in equity markets.

\subsubsection*{1. Metaorders on energy and agricultural futures}
We start by integrating proprietary execution data from Capital Fund Management (CFM), consisting of $12,000$ metaorders across energy and agricultural futures, including oil, gas, metals, corn, wheat, and soybeans traded on the Chicago Mercantile Exchange (CME) from 2012 to 2022. The inherent challenge stems from the disparate trading hours - energy futures exhibit near-continuous 23-hour trading sessions, whereas agricultural products have intraday breaks. To address this, we constrain our analysis to a synchronous trading window from 9:30 AM to 2:20 PM Eastern Time (ET), enabling a consistent interpretation of execution trajectories across the combined dataset.  

\subsubsection*{2. Synthetic metaorders on corn futures}
To delve deeper into the self- and cross-impact dynamics within the agricultural sector under conditions of sparse data, our analysis narrows to corn futures. We select three contracts with staggered maturities, all based on the same underlying and exhibiting return correlations above 90\%. Specifically, we use the front contract (Corn0) and the next two deferred maturities (Corn1 and Corn2), expiring in one, two, and three quarters, respectively. This dataset comprises approximately $\num{1500}$ metaorders. 

In recognizing the limited availability of proprietary metaorders, which restricts reliable impact estimation, we augment this dataset with a synthetic metaorder proxy specifically for corn futures. Inspired by empirical insights and methodologies highlighted in \cite{maitrier2025doublesquarerootlawevidence, sato2024}, this proxy is constructed through an algorithm that introduces trade-sign autocorrelation and randomized trader ID assignment. This synthetic enrichment ensures consistency with observed market microstructure features and enhances the robustness of price impact estimation at the single-asset and portfolio level.  

\subsubsection*{3. Order flow imbalance in equities}
Finally, we use high-frequency tick data from the LOBSTER database to analyze price dynamics driven by aggregate order flow imbalance. The dataset includes detailed order book events for all NASDAQ-traded stocks, in particular, all executed orders. Our study focuses on 197 constituents of the S\&P 500 with complete data across all 237 regular trading days in 2024, where trading hours run from 9:30 AM to 4:00 PM ET.

Unlike the proprietary metaorder data, which enables the study of the price impact of an individual agent’s actions, this dataset allows us to model how prices respond to the total observed order flow in the market. It thus addresses a distinct but complementary question: not how one's own trading moves prices, but how prices change as a result of aggregated traded volume. Understanding and quantifying this mechanism is crucial for the optimal execution of one's own trades as well.

\subsection{Metaorder Proxy}\label{sec:meta-proxy}

The number of available metaorders in the proprietary dataset is insufficient to support a reliable estimation of the impact kernel on a per-product basis. As shown in \cite{neuman2023offline}, the convergence rate of the estimator depends on both the number of samples and the number of trading periods $M$. While the dataset does not permit an increase in $M$, the number of effective samples can be enhanced by introducing a synthetic metaorder generation procedure, inspired by the findings in \cite{maitrier2025doublesquarerootlawevidence}.

The approach originates from proprietary data on the Tokyo Stock Exchange, where the authors had access to full trade-level information, including trader identifiers. They observed that individual traders tend to execute consecutive trades in the same direction, forming buy or sell sequences. Metaorders were identified by aggregating consecutive trades of the same sign per trader, resetting upon sign reversal. The authors of \cite{maitrier2025} noticed that the presence of trader IDs was not essential to reproduce this behavior. Instead, one can simulate such identifiers by randomly assigning trades to synthetic IDs, thereby generating proxy metaorders via the following algorithm:
\begin{enumerate} 
\item Assign to each trade a random integer $n_T \in \{0,1,\ldots,N_T-1\}$ where $N_T$ is a pre-determined fixed positive integer. The integers are sampled uniformly from this finite set, representing synthetic trader IDs.
\item Group all trades with the same assigned $n_T$ and sort them in chronological order. 
\item Partition each sequence into metaorders by aggregating trades with identical signs (buy or sell), terminating the sequence upon sign change. \end{enumerate}

Table~\ref{table:ex_metaproxy} illustrates this procedure on a toy example. For instance, trades with $n_T = 0$ form a consistent sell-side sequence $\{-1,-1,-1 \}$, corresponding to a single metaorder with three child orders. In contrast, trades with $n_T = 1$ exhibit a sign change, resulting in two separate metaorders with one and two child orders, respectively. The average length of a metaorder is influenced by both the chosen value for $N_T$ and the empirical rate of sign changes in the underlying trade flow.

Figure~\ref{fig:meta_peak} displays the peak impact, as defined in \eqref{eq:peak}, for synthetic metaorders on TBOND, 10USNOTE, EUROSTOXX, DAX, and CORN0 futures, filtered to include only those with at least four child orders. All observed impact curves exhibit concave scaling behavior, consistent with a power-law relation. Empirically, the impact function fits a square-root law with exponent in the range $\dl\in [0.5,0.7]$. The metaorder volume range is bounded below by microstructural properties (notably for large-tick assets such as DAX) and above by the decreasing likelihood of long, uninterrupted sign sequences in the trade flow. In practice, the upper bound is observed around 1\% of the daily traded volume.

\begin{table}[H]
\centering
\begin{tabular}{llll}
\toprule
Time & Volume Traded & Trader ID ($n_T$) & Metaorder Assignment \\
\midrule
10:05:011  & -1 & 0 & 1 \\
10:06:123  & -1 & 1 & 2\\
10:06:509  & -1 & 2 & 3\\
10:07:205  & -1 & 0 & 1\\
10:07:388  & 1 & 2 & 4\\
10:07:434  & 1 & 3 & 5\\
10:07:786  & -1 & 1 & 2\\
10:08:657 & -1 & 3 & 6\\
10:09:476  & -1 & 0 & 1\\
10:09:567  & 1 & 1 & 7\\
\bottomrule
\end{tabular}
\caption{Example of synthetic metaorder construction via randomized trader ID assignment and sign-based grouping.}
\end{table}\label{table:ex_metaproxy}

\begin{figure}[H] \centering \includegraphics[width=0.85\linewidth]{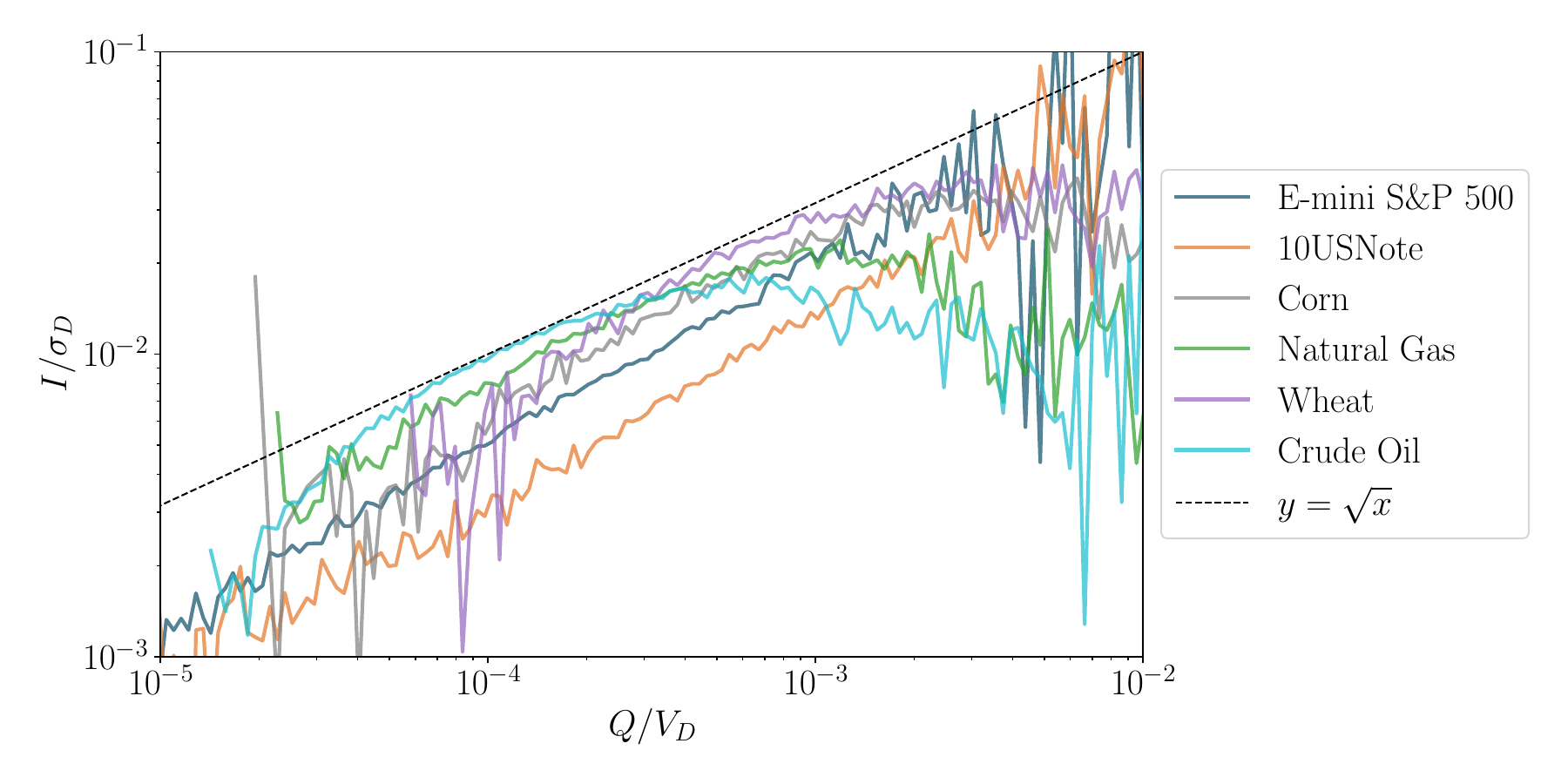} \caption{For various futures contracts, the synthetic meta-order impact rescaled by daily volatility $I/\sigma_D$ is plotted against its volume fraction $Q/V_D$ in log-log scale. The dashed black curve represents square-root scaling. The characteristic concave shape of the impact function is reproduced and holds broadly across assets.} \label{fig:meta_peak} \end{figure}

Choosing $N_T$ requires careful calibration to the microstructural properties of the asset under study. For small $N_T$ (e.g., $N_T \approx 1$), artificially long metaorders are generated, which are unrealistic in practice. Conversely, large $N_T$ (e.g., $N_T \approx $ number of trades per day) results in metaorders with trades sparsely distributed over the day, diluting the temporal structure. Thus, intermediate values of $N_T$ must be selected to capture liquidity and volatility features of the asset.

The minimum and maximum feasible metaorder sizes depend on both the asset's tick size and its intraday trade frequency. Figure \ref{fig:proxy_tick} demonstrates that the smallest possible metaorder volume is influenced by the average trade size in the underlying asset. For large-tick assets, where the average trade size is relatively small, fewer shares are needed to achieve a significant dollar position, resulting in smaller metaorders. This suggests an inverse relationship between the minimum volume of metaorders and the average trade size of the asset. Consequently, the computational thresholds, which vary between $10^{-6}$ and $10^{-4}$, as shown in Figure \ref{fig:meta_peak}, can be explained by the average tick size. Assets such as the DAX and CORN contracts exhibit larger thresholds compared to their smaller tick counterparts.

\begin{figure}
	\center
		\includegraphics[width=0.6\linewidth]{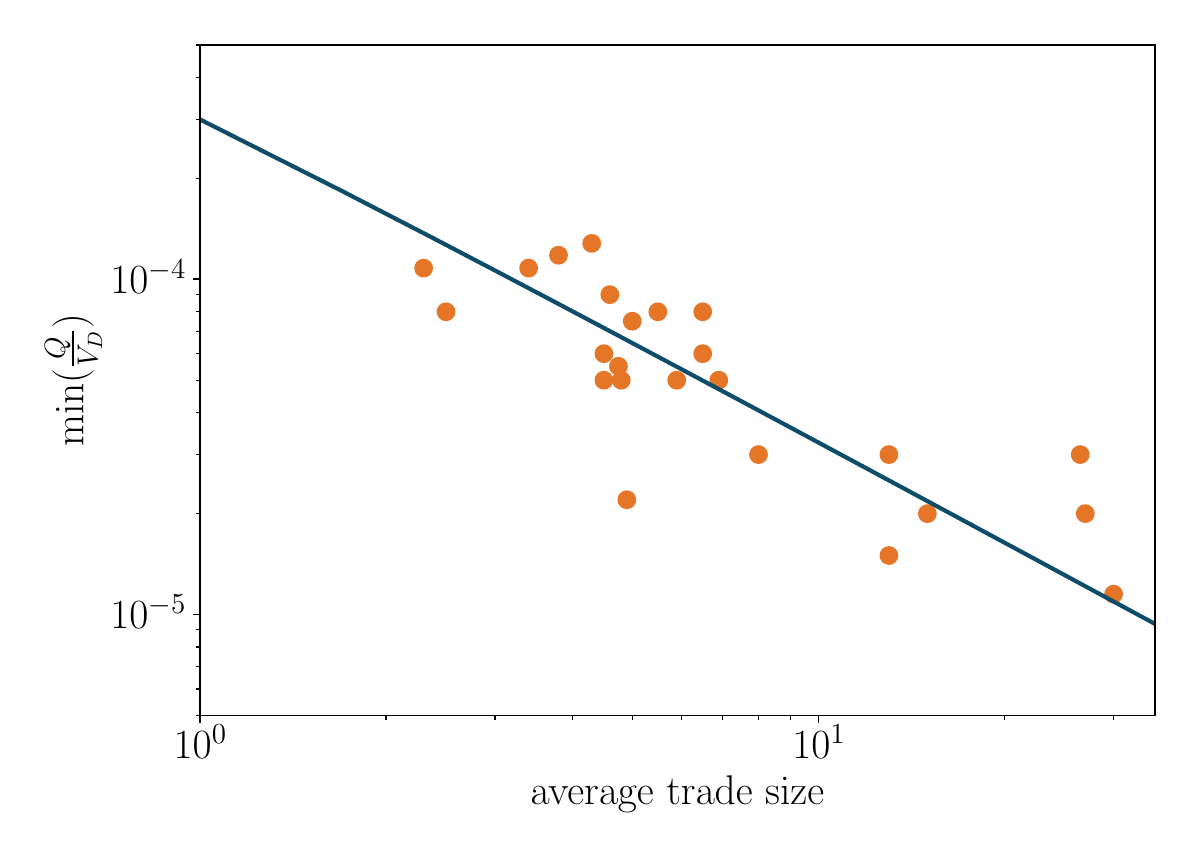}
\caption{
Minimum metaorder size as a function of the average trade size in the underlying asset. The metaorder volume depends on the tick-size of the underlying asset; large ticks require less trades to accumulate a large volume.}\label{fig:proxy_tick}
\end{figure}

\subsection{Methodology}\label{sec:methods}

\subsubsection*{Metaorders on futures}
All metaorder-based analyses, whether from proprietary execution data or synthetic proxies, are conducted on a uniform 5-minute time grid.  Even though some futures on the CME are traded 23 hours a day, we restrict ourselves to the most liquid trading hours that begin with the agricultural futures market opening at 9:30 AM ET and extends through 59 subsequent 5-minute bins, concluding at 1:20 PM. Return observations are constructed from this reference point. The corresponding volume process is set to zero in time bins where no trades occur. Assuming a Time-Weighted Average Price (TWAP) strategy, volumes during execution are presumed to be evenly distributed across the 5-minute bins.

To estimate cross-impact across multiple instruments, we analyze specific pairs and triplets of futures contracts that are have a correlation of more than 90\%. This involves systematically pairing each front-month contract with its subsequent contract. For contracts nearing expiration, we exclude trading activity during the settlement period to avoid the peculiarities associated with end-of-life contract dynamics. This practice helps maintain the accuracy of cross-impact calibration.


\subsubsection*{Order flow imbalance in equities}
All analyses based on order flow imbalance are carried out on an equidistant 10-second time grid. We partition the 6.5-hour trading session of the S\&P 500, spanning from 9:30 AM to 4:00 PM ET, into 2340 distinct 10-second bins. This consistent partitioning aligns with established research \cite{cont2014price,muhle2024stochasticcross,muhle2024stochastic}, balancing the need for precision with data and absence of microstructure effects. The order flow imbalance is quantified as the cumulative sum of signed trade volumes within each bin. 

From the original pool of 503 S\&P 500 stocks as of December 31, 2024, we focus exclusively on constituents exhibiting active limit order book presence in at least 1800 bins per day. This condition narrows our analysis to 197 stocks. We exclude the final half-hour of the trading day to avoid the heightened volatility seen then, but retain the opening hour since price impact builds up most strongly during this period and is therefore essential for intraday decay estimation, resulting in a refined data set of 2160 bins per day.

To estimate cross-impact across multiple stocks, we analyze the stock pair Coca-Cola (KO), PepsiCo (PEP) and the stock triplet triplet ConocoPhillips (COP), Chevron (CVX), Exxon Mobil (XOM), which exhibit pairwise 10-second return correlations of 50-70\% throughout 2024. In contrast to the case of metaorders on futures, we don't have to account for expiries, settlements, or other contract dynamics here. 

In line with \cite{muhle2024stochasticcross}, to assess the predictive power when introducing cross-impact, we consider a market portfolio and assess the prediction  performance of a bivariate cross-impact model which captures cross-impact between the market portfolio and the stocks, instead of focusing on specific stock pairs. Here the market portfolio is constructed as a weighted sum, with prices and trades computed by weighting prices and trades of the individual assets by traded notional in the respective bin. Subsequently, the market portfolio is treated like an asset for the model calibration; see \cite{muhle2024stochasticcross} for details.

\subsubsection*{Normalization}
To correct for intraday volatility patterns in all datasets, such as higher volatility at the open and close, returns are scaled by a modified Garman–Klass volatility estimator. Recalling Definition~\ref{def:dataset}, given an asset $\ell\in\{1,\ldots d\}$ and a time interval $[t_j,t_{i+1}]$ for $0\leq j\leq i\leq M-1$, consider
$$
\sigma^{\ell}_{[t_j,t_{i+1}]} := \frac{1}{3}(\text{High} - \text{Low})_{[t_j,t_{i+1}]}^{\ell} + \frac{2}{3}|\text{Last} - \text{First}|_{[t_j,t_{i+1}]}^{\ell},
$$
where "High" and "Low" represent the highest and lowest price within the respective interval and "First" and "Last" represent the mid-prices at the beginning and end of the interval. As bin durations increase, volatility scales with the square root of time. This becomes apparent from the first subplot in Figure~\ref{fig:intraday_patterns}, where the two curves represent the 10-second bin volatility and the cumulative volatility of AAPL, respectively.

As for volatility, trading volume exhibits a comparable intraday pattern. The lower panel in Figure~\ref{fig:intraday_patterns} highlights the traded volume fraction specifically for futures contracts. Notably, there are peaks in trading activity associated with predictable news events, such as the Weekly Petroleum Status Report and the World Agricultural Supply and Demand Estimation (WASDE) Report. These peaks demonstrate the significant influence of scheduled announcements on liquidity patterns throughout the trading day.

During periods of high liquidity, the same trade volume results in a smaller price impact compared to periods of low liquidity. To adjust for this variance, trade volumes are normalized using the average volume profile corresponding to each time of day. Specifically, let $V_D$ denote the total daily traded volume and $V_{t_i} $ the volume in a given bin $[t_i,t_{i+1}]$. Then, we normalize the volume $(Q_{t_i}^{\ell})^{(n)}$ in an asset $\ell\in\{1,\ldots,d\}$ at timestep $t_i\in\mathbb{T}$ in episode $n\in\{1,\ldots ,N\}$ as follows:
$$(\tilde{Q}_{t_i}^{\ell})^{(n)} := \frac{(Q_{t_i}^{\ell})^{(n)}}{(V_D^{\ell})^{(n)}} \frac{1}{w}\sum_{r = n-w+1}^n \frac{(V^{\ell}_D)^{(r)}}{(V^{\ell}_{t_i})^{(r)}},$$
where samples are chronologically ordered in time and the rolling window mean is defined as $w: = \text{min}(n,20)$ days. To account for different concavities, we recall the class of impact functions $h_c$ from \eqref{eq:hc}, and introduce the notations $h_{c_S}$ and $h_{c_X}$ to differentiate between self- and cross-impact, respectively. 
Taking into account the aforementioned normalization and specification of the impact function, the propagator model \eqref{eq:returns} takes the form

\begin{equation}\label{eq:propagator_normalized}
    \frac{P^{\ell}_{t_{i+1}} - P^{\ell}_{t_0}}{\sigma^{\ell}_{[t_0,t_{i+1}]}} = \sum_{j=0}^i G^{({\ell},{\ell})}_{i-j}h_{c_S}\Big(\tilde{Q}^{\ell}_{t_j}\Big) + \sum_{k\neq \ell}^d\sum_{j=0}^i G^{(\ell,k)}_{i-j}h_{c_X}\Big(\tilde{Q}^k_{t_j}\Big),\quad i = 0,\ldots,M-1,\ \ell=1,\ldots d,
\end{equation}
where we omit the superscript $(n)$ for simplicity. Unlike in previous studies such as \cite{Patzelt_2017,taranto2016linearmodelsimpactorder}, the kernel $\boldsymbol{G}$ in our setup does not absorb the volatility term. This normalization ensures robustness with respect to noisy estimates for small bin sizes, and allows the impact kernel to be estimated from inputs that are free from seasonal intraday patterns in volatility and liquidity. 

\begin{figure}
    \begin{subfigure}{\textwidth}
        \includegraphics[width=0.915\textwidth]{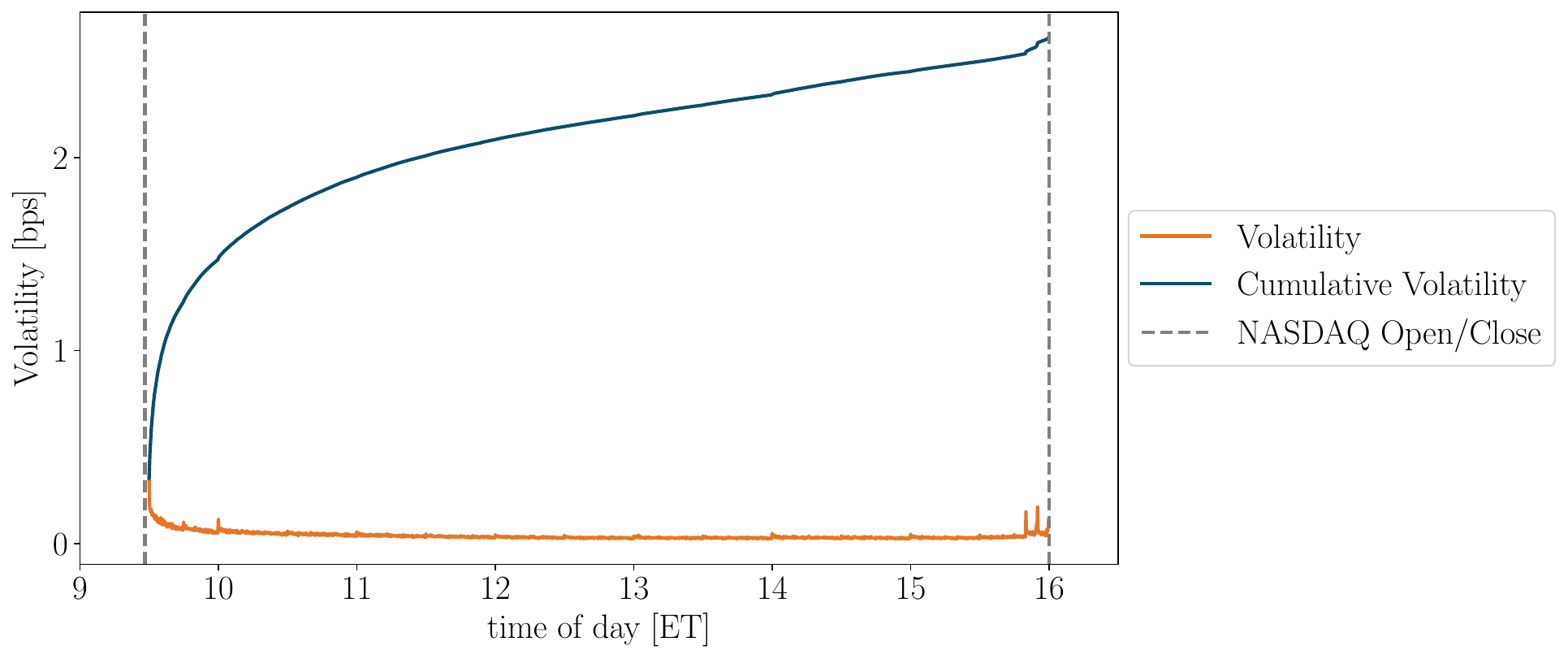}  
    \end{subfigure}
        \begin{subfigure}{\textwidth}\label{subfig:volumepattern}
        \hspace{0\textwidth}
        \includegraphics[width=\textwidth]{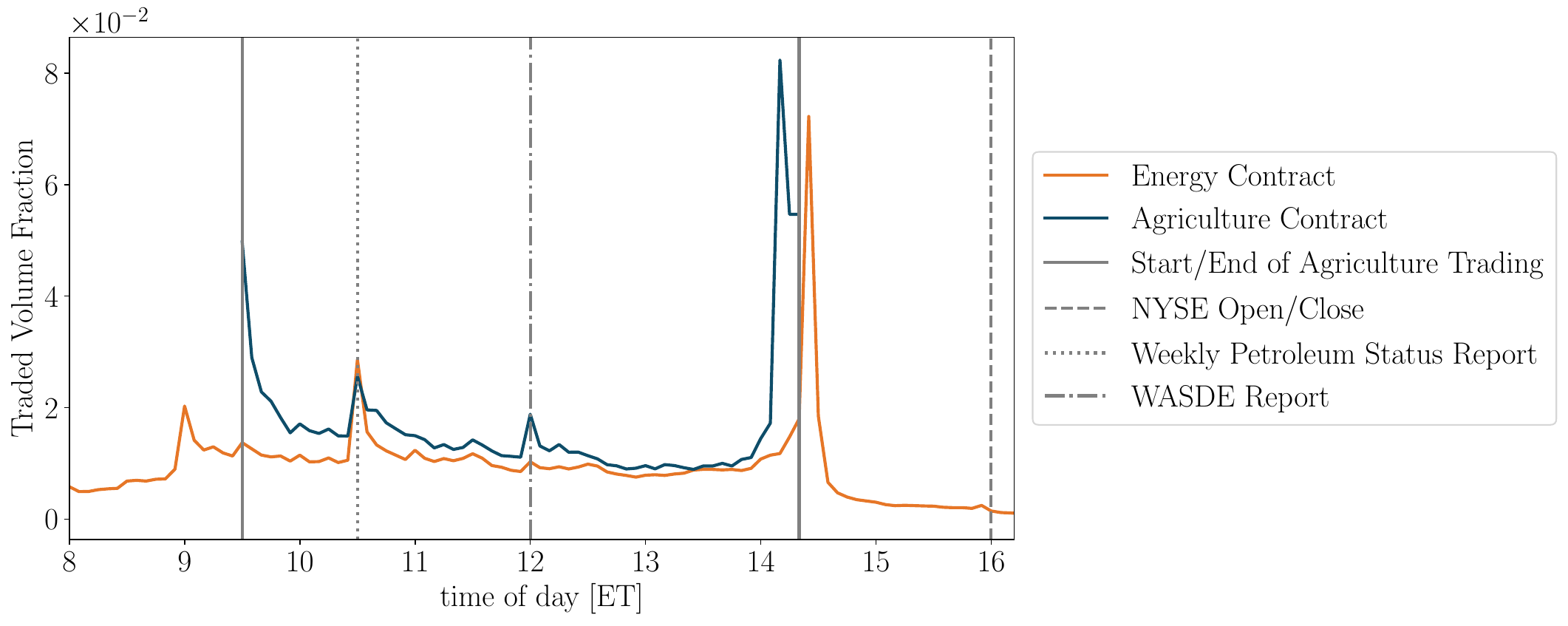}
    \end{subfigure}
        \caption{Upper panel: the volatility and cumulative volatility patterns of AAPL stock throughout NASDAQ trading hours, emphasizing varying impacts over different trading periods. Vertical lines mark the NASDAQ opening and closing times. Lower panel: the intraday traded volume profile for energy (orange) and agriculture (blue) futures contracts. Key events such as the Weekly Petroleum Status Report and the WASDE Report influence trading patterns and are indicated with vertical lines. Vertical lines also indicate trading times allowing for liquidity comparisons at different times of the day. Both volume and volatility normalizations are essential for capturing intraday variations in liquidity and return, with each contract analyzed over a rolling 20-day window.}
    \label{fig:intraday_patterns}
\end{figure}

\subsubsection*{Predictive Power}
To assess the predictive power of the estimated impact kernels, we employ different strategies tailored to each dataset type.

For the metaorder datasets (both real and synthetic), we estimate the kernel on the full sample rather than splitting it into separate training and test sets, given the limited data available. This choice ensures that calibration incorporates all observed metaorders, providing the most reliable estimates in this context. Model performance is then evaluated through the coefficient of determination $R^2$, which allows for direct comparison across specifications.

In contrast, for the order flow datasets, we evaluate the explanatory strength of the impact kernels by examining both in-sample and out-of-sample predictive performance. We adopt a rolling time-based evaluation scheme similar to that in \cite{muhle2024stochasticcross}. Here, models are trained on one month of data and validated on the subsequent month, repeated over a one-year span. The model performance is then measured across the two prediction horizons $h = 1$ minute and $h = 5$ minutes, that is, assuming the traded volumes are known, prices are predicted 1 minute or 5 minutes in advance using the estimated models. The resulting in-sample $R^2$ values are averaged over 12 months, and the out-of-sample $R^2$ values are averaged over 11 months.

\subsubsection*{Price Impact Models}
We compare the performance of the nonparametric estimator with classical parametric propagator models, where all models share the general form \eqref{eq:propagator_normalized}. Given an asset pair $(\ell, k) \in \{1,d\}^2$ and $0\leq j\leq i\leq M-1$, we consider the parametric kernels
\begin{itemize}
\item \textbf{Single-Exponential (1-EXP):} $G_{i-j}^{(\ell,k)} = Y e^{-\rho(t_{i}-t_j)}$;
\item \textbf{Double-Exponential (2-EXP):} $G_{i-j}^{(\ell,k)} = Y \Big(w_1 e^{-\rho_1(t_{i}-t_j)} + (1-w_1)e^{-\rho_2(t_i-t_j)}\Big)$;
\item \textbf{Power-Law (POWER):} $G_{i-j}^{(\ell,k)} = Y (t_i-t_j + \tau)^{-\beta}$;
\end{itemize}
where $Y$ is a constant of order one as in \eqref{eq:peak}, and the two nonparametric kernels
\begin{itemize}
\item \textbf{Raw Nonparametric (RAW):} $G_{i-j}^{(\ell,k)} =(\tilde{G}_{N,\lambda})_{i-j}^{(\ell,k)}$, as defined in \eqref{eq:unconstrained};
\item \textbf{Projected Nonparametric (PROJ):} $G_{i-j}^{(\ell,k)} =(G_{N,\lambda})_{i-j}^{(\ell,k)}$, as defined in \eqref{eq:constrained}.
\end{itemize}
In line with \cite{hey2023concave,muhle2024stochasticcross,muhle2024stochastic}, each of the three parametric models is fitted by fixing $\rho$, $(\rho_1,\rho_2)$, or $(\tau,\beta)$ first, and precomputing the corresponding normalized impact changes $(I_{t_{i+1}}-I_{t_i})/Y$.
Then, assuming that alpha signals average out, observed price changes can be decomposed into price impact and unaffected prices, so that $Y$ can be estimated by regressing price changes against the normalized impact changes. Finally, parameter tuning is performed via a grid search to find the optimal values $\rho^*$, $(\rho^*_1,\rho^*_2)$, or $(\tau^*,\beta^*)$. Specifically, for the grid search of the exponential models, we define the interval of half-lives $\log(\num{2})/\rho \in [30,\num{720000}]$ seconds and denote by $N_\rho$, $N_{\rho_1}$, $N_{\rho_2}$ the numbers of corresponding grid points within this interval. For the power-law models, the grid search includes decay parameters $\beta\in [0,1]$ and shifts $\tau \in [10, \num{7200}]$ seconds, with $N_{\beta}$, $N_{\tau}$ specifying the numbers of associated grid points. Table~\ref{tab:model-parameter-comparison} summarizes the degrees of freedom and tuning requirements for each model. The complexity of parametric models increases significantly as the number of grid points grows, leading to a rapid escalation in computational demands for parameter optimization.

\begin{table}[htb]
\centering
\begin{tabular}{||c|c|c|c||}
\hline
Model & Free parameters & Tuned parameters &  Grid points \\
\hline\hline
1-EXP & $Y,\rho$ & $\rho$  &$N_{\rho}$\\
2-EXP & $Y, w_1,\rho_1,\rho_2$ & $\rho_1,\rho_2$ &  $N_{\rho_1}N_{\rho_2}$\\
POWER & $Y,\beta,\tau$ & $\beta,\tau$ &  $N_{\beta}N_{\tau}$ \\
RAW & $(\tilde{G}_{i})_{i=0}^{M-1}$ & None  & 1\\
PROJ & $(G_{i})_{i=0}^{M-1}$ & None  & 1\\
\hline
\end{tabular}
\caption{The free parameters per propagator, tuned parameters per propagator, and the number of grid points per propagator for the different models. Here, $N_{\rho}$, $N_{\rho_1}$, $N_{\rho_2}$, $N_{\beta}$, and $N_{\tau}$ refer to the number of grid points for $\rho$, $\rho_1$, $\rho_2$, $\beta$, and $\tau$, respectively.} 
\label{tab:model-parameter-comparison}
\end{table}

\subsection{Understanding Metaorder Impact}\label{sec:results}
This section examines the impact of metaorders across different model specifications, as detailed in Tables \ref{tab:metaorder-model-comparison}, \ref{tab:metaorder-crossimpact-comparison} and \ref{tab:R2}. Our key findings are the following: 

\begin{enumerate}
    \item Self-impact is shown to be concave and to decay across multiple timescales (a relative $R^2$ improvement of 
    $94.4\%$ from a linear to a concave model and $1.06\%$ for the nonparametric estimator as compared to exponential decay). The nonparametric estimators supports power-law or multi-exponential decay without assuming any specific kernel shape, as stated in Table~\ref{tab:metaorder-model-comparison} and illustrated in Figure~\ref{fig:cfm_enhanced_single}.
    \item The metaorder proxy enhances decay calibration when data is sparse, as can be seen in Table \ref{tab:R2} (a relative $R^2$ improvement of $14.29\%$). Data enhancement offers an improved predictive accuracy by providing smoother and more reliable decay profiles. 
    \item Table \ref{tab:metaorder-crossimpact-comparison} highlights that cross-impact is concave (a relative $R^2$ improvement of $1.74\%$ from linear to concave cross-impact). The predictive power also increases with the number of assets that are considered, as shown in Table \ref{tab:R2} (a relative $R^2$ improvement of $31.25\%$ for one cross-asset and $35.45\%$ for two cross-assets). Cross-impact can reflect combined effects of self-impact responses and order flow imbalances. Figure~\ref{fig:R2_concave} illustrates the models performance ratio when varying the cross-impact concavity parameter. 
\end{enumerate}
In the following two subsections, we provide additional details about our results on self- and cross-impact estimations.

\subsubsection{Self-Impact Nonparametric Estimation} 
We begin by evaluating self-impact using the original CFM metaorder dataset, which contains trades on energy and agricultural futures contracts without synthetic metaorders. A key element of this analysis is the concavity of the self-impact function $h_{c_S}$ from \eqref{eq:propagator_normalized}, parameterized by $c_S$. For each value of $c_S$, we compute the coefficient of determination $R^2(c_S)$ to assess predictive accuracy. The point estimate $\hat c_S$ denotes the value of $c_S$ that maximizes $R^2(c_S)$. Figure~\ref{fig:R2_concave} shows the relative performance ratio $R^2(c_S)/R^2(\hat c_S)$ for the single-asset model (blue). The curve peaks around the square-root case ($c_S \approx 0.5$). The value significantly decreases for higher or lower values in line with the results of \cite{hey2023cost}. 

\begin{table}[htb]
\centering
\begin{tabular}{||c|c|c||}
\hline
Propagator & Concavity $c_S$ & $R^2$ \\
\hline
\hline
1-EXP & 1 &0.53\% \\
2-EXP & 1 &0.54\%  \\
POWER & 1 &0.53\%  \\
RAW & 1 &0.54\%  \\
PROJ & 1 &0.54\% \\
\hline
1-EXP & 0.5 &1.036\% \\
2-EXP & 0.5 &1.049\% \\
POWER & 0.5 &1.044\%  \\
RAW & 0.5 &1.05\% \\
PROJ & 0.5 &1.049\%  \\
\hline
\end{tabular}
\caption{Performance of linear and square-root propagator models on the original CFM metaorder dataset. Linear models perform approximately half as well. Models with multiple transient impact decay timescales perform best.}
\label{tab:metaorder-model-comparison}
\end{table}

We then compare various decay kernel specifications under linear and square-root impact. Table~\ref{tab:metaorder-model-comparison} summarizes the $R^2$ values in \% of explained variance for all considered models. Among these, the nonparametric kernel (without projection) achieves the highest performance of approximately 1.05\% as expected. Notably, the projected nonparametric kernel performs equally well as the parametric multi-timescale models, despite requiring no tuning.

Overall, the relatively modest $R^2$ values align with expectations given CFM's limited trading volume relative to the broader market. The measured values highlight the subtler influence of CFM's trading activities, falling within the anticipated range of $0.1-1$\% that corresponds to the average volume fraction traded in the time interval.

\begin{figure}
    \centering
    \includegraphics[width=0.7\linewidth]{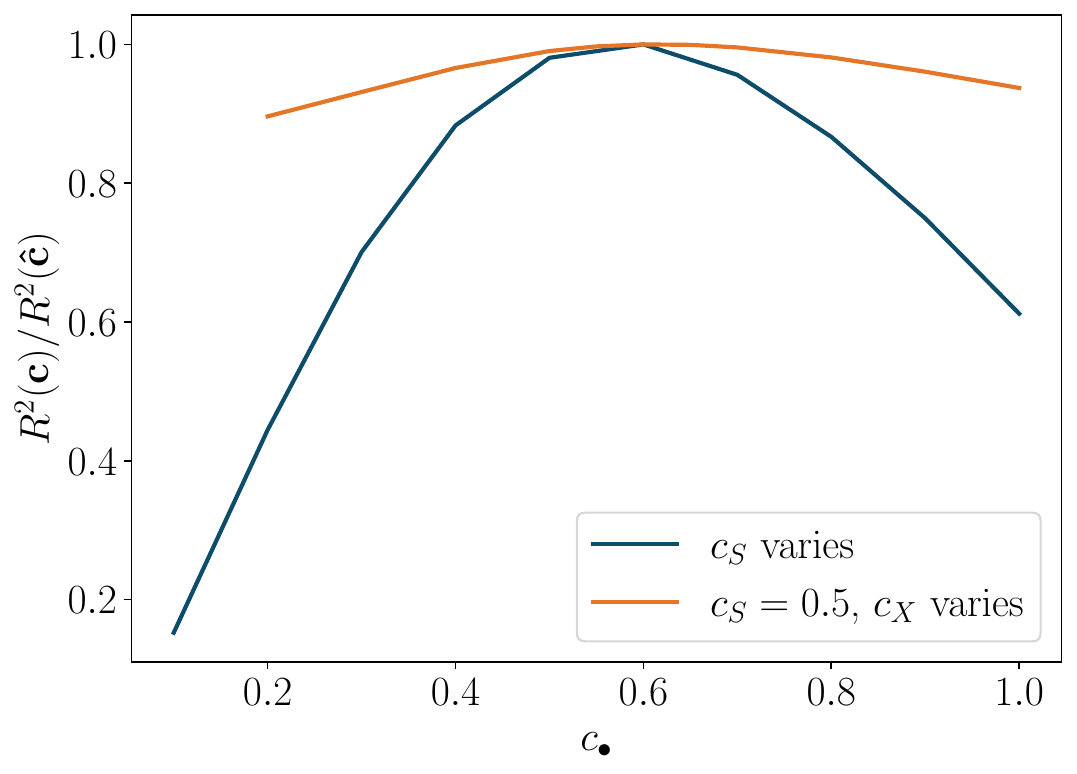}
    \caption{$R^2$ performance as a function of the concavities $c_\bullet$ for self- and cross-impact. The blue curve represents $R^2$ for self-impact models with varying concavity parameters $c_S$. The orange curve shows $R^2$ for cross-impact models, fixing $c_S = 0.5$ while varying $c_X$. Both curves peak around $0.5-0.6$.}
    \label{fig:R2_concave}
\end{figure}

\subsubsection{Cross-Impact Estimation}

We now extend the analysis to cross-impact using metaorders on multiple contracts with different expiries. The goal is to understand how the execution of one asset affects the price of another and whether such effects can improve prediction performance beyond self-impact alone.

To evaluate the predictive value of cross-impact, we compare two classes of models. First, we consider the self-impact-only models RAW and PROJ, which estimate the kernels based on the traded volume of each asset individually. We recall the notions of raw and projected estimators arise from \eqref{eq:unconstrained} and \eqref{eq:constrained-rewritten} respectively. Specifically, we calibrate the self-impact kernel components in the first term of \eqref{eq:propagator_normalized} for each asset $\ell\in\{1,\ldots, d\}$, and set the cross-impact kernel components in the second term of \eqref{eq:propagator_normalized} to zero. Second, we consider the cross-impact models CROSS-RAW and CROSS-PROJ, incorporating both self-impact and cross-impact by accounting for influences from contracts with the next expiry. In these models, we calibrate all kernel components in \eqref{eq:propagator_normalized} 
across all asset pairs $(\ell,k)\in\{1,\ldots d\}^2$. 

\begin{table}[htb]
\centering
\begin{tabular}{||c|c|c||}
\hline
Propagator & Concavity $c_X$ & $R^2$ \\
\hline
\hline
RAW & -- & 1.05\%  \\
PROJ & -- & 1.049\%  \\
CROSS-RAW  &1 & 1.15\% \\
CROSS-PROJ & 1 & 1.13\%  \\
CROSS-RAW & 0.5 & 1.17\% \\
CROSS-PROJ & 0.5 & 1.15\%  \\
\hline
\end{tabular}
\caption{Performance of self and cross-impact models on metaorder data using projected (PROJ) and raw (RAW) kernels. The concave cross-impact model fits best.}
\label{tab:metaorder-crossimpact-comparison}
\end{table}

Table~\ref{tab:metaorder-crossimpact-comparison} summarizes the $R^2$ values for these models, using both a square-root and linear function $h_{c_X}$ for the cross-impact term while keeping $c_S = 0.5$. The projected square-root cross-impact kernel improves $R^2$ from 1.05\% (self-impact only, where $c_X$ is not needed) to 1.17\%. Interestingly, when the cross-impact kernel is assumed to be linear, performance decreases to 1.15\%, showing that cross-impact is concave as suggested in \cite{hey2024cross}. The concavity of cross-impact is further investigated in Figure \ref{fig:R2_concave} where the orange curve plots the models' performance ratio across various concavity parameters $c_X$. When integrating both self- and cross-impact contributions, peaks are observed around $0.5-0.6$.

Nonetheless, accurately measuring $c_X$ is challenging due to cross-impact being a second-order effect. While data suggests that both concavities peak in the same range, determining whether $c_X > c_S$ remains elusive given the low signal-to-noise ratio. Cross-asset price adjustments might occur not only in response to price changes initiated by self-impact but also due to observed order flow imbalances. This suggests that the cross-impact relationship could exhibit characteristics of both square-root and linear effects, reflecting a convolution of these dynamics.

\subsubsection{Impact Estimation with the Enhanced Dataset}

Figure~\ref{fig:cfm_enhanced_single} compares the impact kernel estimated from the trades on corn futures (gray) with that from the proxy-enhanced dataset (blue). The kernel from raw CFM metaorders appears irregular and noisy, whereas the extended sample yields a smooth, monotonic decay consistent with known empirical regularities. Figure~\ref{fig:Kernel1dloglog} further illustrates this in a log-log scale. After rescaling the kernel by a factor of $\sqrt{t}$ to correct for volatility normalization in the returns, it follows a near-linear trend with a slope of $-0.5$, consistent with the square-root law observed in \cite{Bouchaud2004,Bouchaud2015}.

While concavity in volume is the dominant contribution to the predictive power of the impact-model, the decay is a second-order effect, consistent with findings in \cite{hey2023concave}. The first $3$ rows of Table~\ref{tab:R2} show that $R^2$ improves when including proxy data.

\begin{table}[htb]
\centering
\begin{tabular}{||c|c|c|c|c||}

\hline
Concavity $c_S$ & Concavity $c_X$ & Assets & Enhanced & $R^2$ \\
\hline
\hline
1 & -- & 1 & \xmark & 3.2\% \\
0.5 & -- & 1 & \xmark & 4.6\% \\
0.5 & -- & 1 & \cmark & 4.8\% \\
\hline
0.5 & 0.5 & 2 & \cmark & 6.3\% \\
0.5 & 0.5 & 3 & \cmark & 6.5\% \\
\hline
\end{tabular}
\caption{Performance of self and cross-impact models  using the nonparametric raw estimator (RAW). The table compares models with different numbers of assets, concavity parameters, and indicates whether the dataset is enhanced with synthetic metaorders.}
\label{tab:R2}
\end{table}

\begin{figure}[ht]
	\center
		\includegraphics[width=0.6\linewidth]{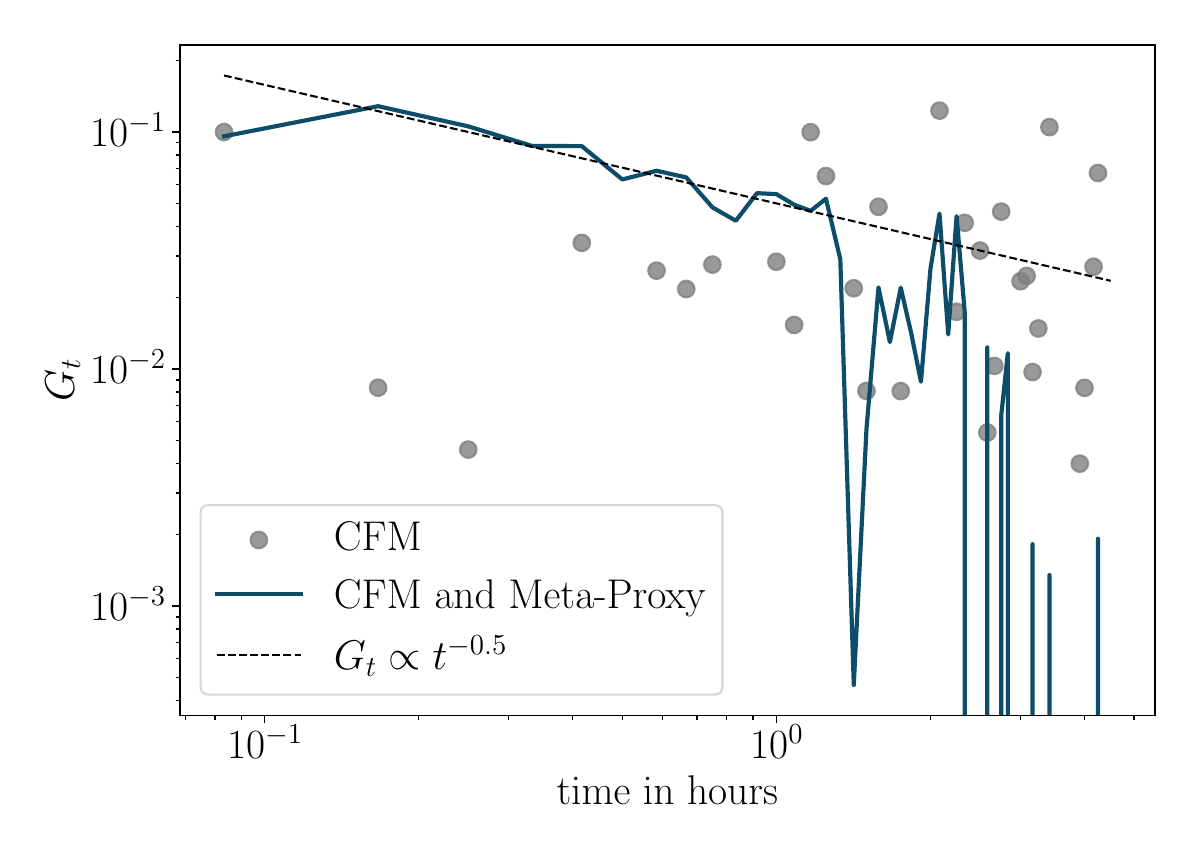}
\caption{Impact kernel estimates $G:=G^{(\ell,\ell)}$ for square-root self-impact with (blue) and without proxy-enhanced date (grey) in log-log scale. The proxy-enhanced kernel, after rescaling, decays smoothly and can be approximated by a power-law propagator with an exponent of -0.5. }\label{fig:Kernel1dloglog}
\end{figure}

Next, we consider nonparametric estimations for multi-asset cross-impact models on the enhanced dataset of corn futures. The two-asset model shows that including a second highly correlated contract increases $R^2$ from 4.8\% to 6.3\%, a substantial gain over the self-impact-only specification. The three-asset model increases $R^2$ further to 6.5\% (see the last two entries of Table~\ref{tab:R2}).

The two left panels in Figure~\ref{fig:MultiAsset} illustrate the estimated propagators for two corn futures with different expirations, indexed by 0 and 1 (see Section~\ref{sec:dataset} for details on the index notation). We observe that the kernels are asymmetric: the impact of Corn0 on Corn1 is stronger than the reverse. This asymmetry reflects liquidity differences across assets, with Corn1 being less liquid. Our findings are consistent with previous observations on immediate aggregate order flow impact in \cite{coz2023cross}. However, while that study focuses on estimating the peak cross-impact kernel as a static object varying with liquidity differences, we go further by directly calibrating the full propagator of metaorders, which captures the dynamics of liquidity differences over time.

The right-hand side of Figure~\ref{fig:MultiAsset} extends the estimation to three assets. While the qualitative patterns persist, the estimates become less stable. The width of the confidence intervals derived in Theorem~\ref{thm:main} increases substantially for a single asset to a multi-asset model. This confirms that cross-impact estimation is subject to high variance, but the structure of the kernels remains interpretable. In particular, cross-impact can exceed self-impact in magnitude, depending on relative liquidity.

\begin{figure}[!htb]
    \centering
    \includegraphics[width=\linewidth]{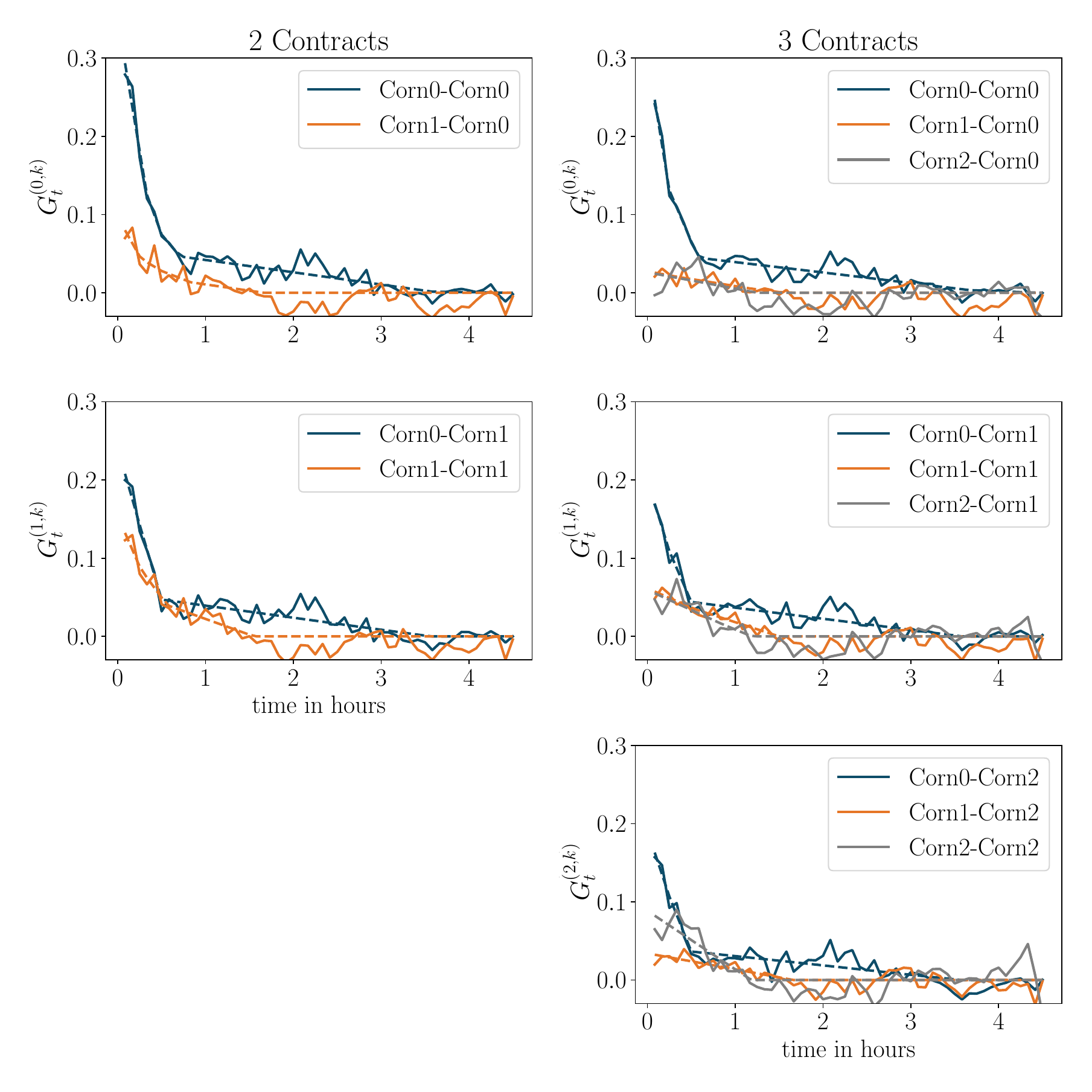}
\caption{Cross-impact kernels are estimated for corn future contracts with different expiries. The solid and dashed lines depict the raw and projected kernels, respectively. Left column: results are shown for a two-asset model, where cross-impact is estimated from Corn0 (blue) to itself (above) or to Corn1 (below) and from Corn1 (orange) to Corn0 (above) or to itself (below). The different cross and self-impact ratios reflect liquidity differences. Right column: for the three-asset model, cross and self-impact kernels are similarly estimated across Corn0, Corn1, and Corn2 (gray). The relative importance of these kernels varies with liquidity, and increasing dimensionality leads to less stable calibrations.} 
    \label{fig:MultiAsset}
\end{figure}

\begin{figure}[!htbp]
    \centering
    \includegraphics[width=\textwidth]{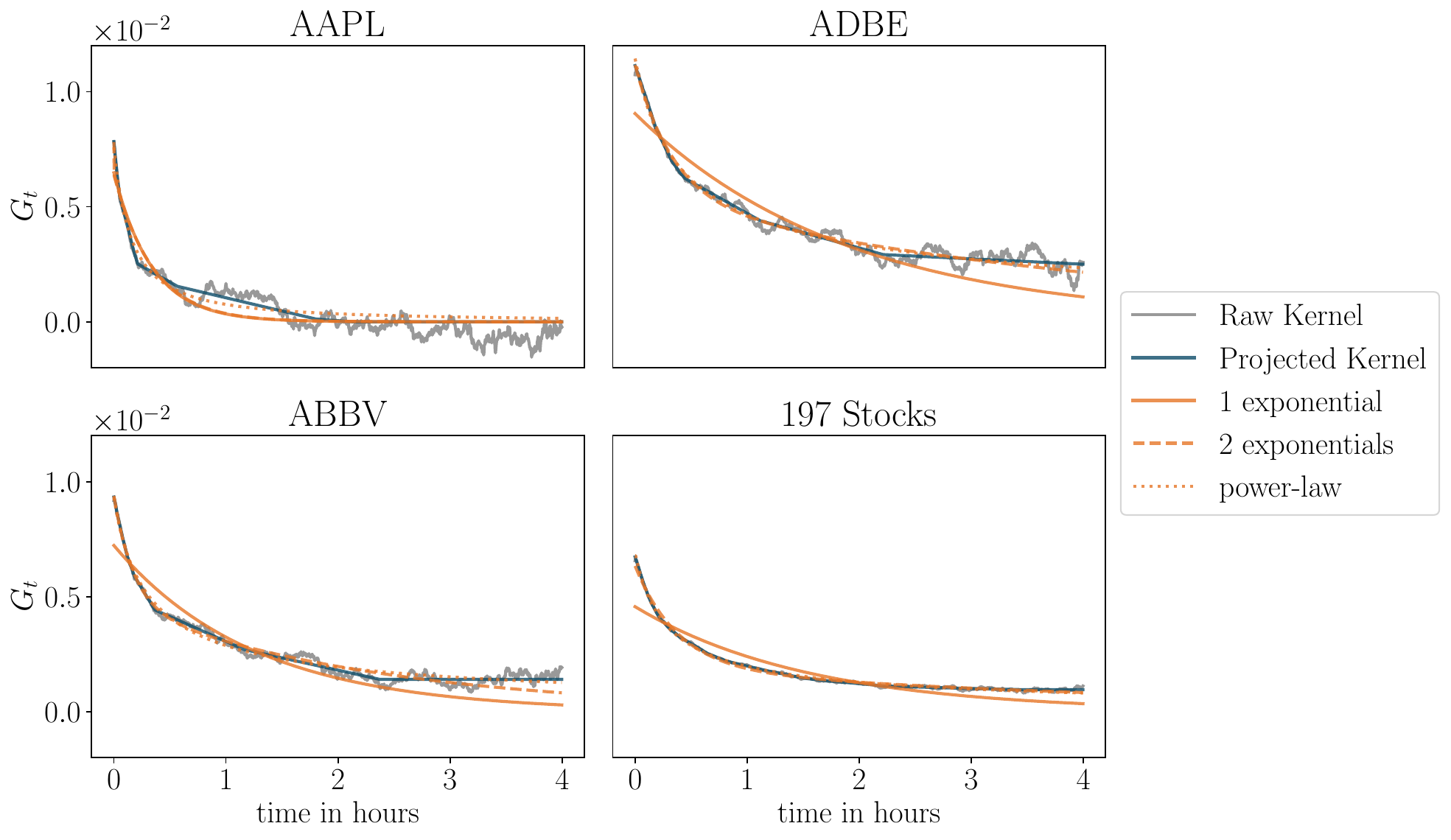}
    \caption{Comparison of impact decay kernels across selected stocks using various estimators. The subplots display the estimated kernel shapes for AAPL, ADBE, and ABBV, alongside an aggregated analysis across multiple stocks. Each subplot includes the raw kernel, projected kernel, and parametric fits using one exponential, two exponentials, and a power-law.}
    \label{fig:kernels}
\end{figure}

\subsection{Understanding Aggregate Impact}\label{sec:empirics-agg}
We now turn our attention to the analysis based on order flow imbalance. Our main results using public market data of S\&P 500 stocks are as follows:

\begin{enumerate}
\item Self-impact is concave (a relative out-of-sample $R^2$ improvement of $75.98\%-183.89\%$ depending on the model and prediction horizon for $c_S=0.5$ compared to $c_S=1$, see Table~\ref{tab:model-comparison}. For the projected nonparametric kernel, we find an optimal concavity parameter of about $c_S\approx0.3$ (a relative out-of-sample $R^2$ improvement of $84.39\%-89.05\%$ depending on the prediction horizon compared to the linear case), see Table~\ref{tab:concavity-comparison}. 
\item As for metaorders, the nonparametric estimator validates a power-law decay of self-impact for aggregate order flow in an unbiased manner, that is, without imposing a parametric structure on the decay kernel, as illustrated in Figures~\ref{fig:of_imbalance_single} and \ref{fig:kernels}. In particular, we find that a shifted power-law kernel yields the best fit for the nonparametrically estimated kernel, as shown in Figure \ref{fig:of_imbalance_single}. 
\item Both the raw and projected estimated kernels achieve higher in-sample performance compared to the parametric kernels (a relative in-sample $R^2$ improvement of $0.89\%-18.64\%$ depending on the model, concavity, and prediction horizon). While the raw kernel exhibits lower out-of-sample performance than the parametric models, the projected estimated kernel performs better out-of-sample (a relative out-of-sample $R^2$ improvement of $1.86\%-24.65\%$ depending on the model, concavity, and prediction horizon), as can be seen from Table \ref{tab:model-comparison}.
\item Cross-impact induced by aggregate order flow exhibits concavity just as in the metaorder case. Nonparametric square-root cross-impact models fit better than their linear counterparts (a relative out-of-sample $R^2$ improvement of $0.87\%-2.2\%$ depending on the prediction horizon), as shown in Table~\ref{tab:model-comparison-cross}.
\end{enumerate}

\subsubsection{Self-Impact Estimation}
Figure \ref{fig:of_imbalance_single} shows the nonparametric kernel estimated universally from all 197 stocks in our dataset, confirming the power-law nature of price impact decay for US equities. Figure \ref{fig:kernels} then displays the estimated kernels for single stocks, illustrating the nonparametric estimation and suggesting a power-law decay as well.

\begin{table}[htb]
    \centering
    \begin{tabular}{||c | c | c|c | c||}
    \hline
    Propagator & Concavity $c_S$ &  Horizon & IS $R^2$    & OOS $R^2$ \\
    
    \hline
    \hline
    1-EXP & 1 & 1 & 16.36\% & 15.45\% \\
    2-EXP & 1 & 1 & 17.74\% & 15.80\%  \\
    POWER & 1 & 1 & 18.03\% & 16.27\% \\
    RAW & 1 & 1 &19.41\% & 9.05\% \\
    PROJ & 1 & 1 & 19.02\% & 16.90\% \\
    1-EXP & 0.5 & 1 & 30.28\% & 29.48\%  \\
    2-EXP & 0.5 & 1 & 30.94\% & 29.99\%  \\
    POWER & 0.5 & 1 & 30.90\% & 30.12\%  \\
    RAW & 0.5 & 1 & 31.27\% & 22.51\%\\
    PROJ & 0.5 & 1 & 31.18\% & 30.68\% \\
    \hline
    1-EXP & 1 & 5 & 11.67\% & 10.02\%  \\
    2-EXP & 1 & 5 & 12.18\% & 11.52\% \\
    POWER & 1 & 5 & 12.42\% & 11.69\%\\
    RAW & 1 & 5 & 12.98\% & 5.40\% \\
    PROJ & 1 & 5 & 12.53\%& 12.49\% \\
    1-EXP & 0.5 & 5 & 22.32\% & 20.63\% \\
    2-EXP & 0.5 & 5 & 22.51\% & 20.67\% \\
    POWER & 0.5 & 5 & 22.43\% & 20.59\% \\
    RAW & 0.5 & 5 & 23.75\% & 15.33\% \\
    PROJ & 0.5 & 5 & 23.36\% &21.98\% \\
    \hline
    \end{tabular}
    
    \caption{In- and out-of-sample performance of the linear and square-root parametric models \mbox{1-EXP}, 2-EXP, POWER and the nonparametric models RAW, PROJ for prediction horizons $h=$ 1min, 5min. The projected nonparametrically estimated kernel PROJ achieves the best prediction power.}
    \label{tab:model-comparison}
\end{table}

Table \ref{tab:model-comparison} summarizes the in-~and out-of-sample performance of the linear and square-root propagator model for different decay kernels and prediction horizons. As explained in Section \ref{sec:methods}, this includes a grid search over the tuned parameters listed in Table \ref{tab:model-parameter-comparison} for each parametric propagator. Notably, the projected nonparametrically estimated kernel achieves the best out-of-sample prediction power among the models considered. This is because the raw kernel estimates tend to be quite spiky, so the projection step effectively smoothens out high-frequency noise and substantially reduces overfitting. By contrast, in the metaorder setting  with 5-minute bins the kernel is already estimated over much coarser intervals, so applying the same projection yields only marginal gains in forecast power.

Table \ref{tab:model-comparison} indicates that the impact concavity parameter exerts a stronger influence on the prediction performance than the impact decay. Motivated by this finding, Table~\ref{tab:concavity-comparison} compares the out-of-sample performance of the projected estimated kernel for various concavity values.

Consistent with the results of \cite{muhle2024stochastic}, our analysis suggests that a concavity parameter of approximately 0.3 yields optimal prediction performance. This outcome is likely due to the fact that we aggregated the data into bins and treated the aggregate signed volume in each bin as a single order, thereby producing relatively large order sizes. These findings align with those of \cite{Lillo2003}, who observed a concavity parameter of around 0.5 for small volumes and 0.2 for large volumes in  US stocks.

\begin{table}
    \centering
    \begin{tabular}{||c|c|c|c||}
    \hline
   Concavity $c_S$ &  Horizon & IS $R^2$    & OOS $R^2$ \\
    
    \hline
    \hline
    0.5 & 1 & 31.18\% & 30.68\% \\
  0.4 & 1 & 32.05\% & 31.38\% \\
  0.3 & 1 & 33.42\% & 31.95\% \\
 0.2 & 1 & 32.46\% & 30.81\% \\
0.1 & 1 & 29.01\% & 27.30\% \\
    \hline
0.5 & 5 & 23.36\% &21.98\% \\
0.4 & 5 & 24.11\% &22.74\% \\
0.3 & 5 & 24.65\% &23.03\% \\
   0.2 & 5 & 24.17\% &22.79\% \\
0.1 & 5 & 22.50\% &20.08\% \\
    \hline
    \end{tabular}
    
    \caption{In- and out-of-sample performance of the nonparametric propagator model PROJ for different concavities $c_S=$ 0.1, 0.2, 0.3, 0.4, 0.5 and prediction horizons $h=$ 1min, 5min. A concavity parameter of 0.3 yields the highest prediction performance.}
    \label{tab:concavity-comparison}
\end{table}

\subsubsection{Cross-Impact Estimation}

\begin{table}
    \centering
    \begin{tabular}{||c |c | c | c|c | c ||}
    \hline
    Model & Concavity $c_S$ & Concavity $c_X$ &  Horizon & IS $R^2$    & OOS $R^2$ \\
    
    \hline
    \hline
    PROJ & 1 & - & 1 & 19.02\% & 16.90\% \\
    CROSS-PROJ & 1 & 1 & 1 & 21.33\% & 18.92\% \\
    PROJ & 0.5 & - & 1 & 31.18\% & 30.68\% \\
    CROSS-PROJ & 0.5 & 1 & 1 & 31.32\% & 30.90\% \\
    CROSS-PROJ & 0.5 & 0.5 & 1 & 32.39\% & 31.17\% \\
    \hline
    PROJ & 1 & - & 5 & 12.53\%& 12.49\% \\
    CROSS-PROJ & 1 & 1 & 5 & 14.58\%& 14.55\% \\
    PROJ & 0.5 & - & 5 & 23.36\% &21.98\% \\
    CROSS-PROJ & 0.5 & 1 & 5 & 23.71\% & 22.25\% \\
    CROSS-PROJ & 0.5 & 0.5 & 5 & 24.25\% &22.74\% \\
    \hline
    \end{tabular}
    
    \caption{In- and out-of-sample performance of the linear and square-root propagator models PROJ (without cross-impact) and CROSS-PROJ (with cross-impact) for horizons $h=$ 1min, 5min. The introduction of cross-impact increases predictive power, with the increase being larger for concave cross-impact.}
    \label{tab:model-comparison-cross}
\end{table}

In line with the analysis in Section \ref{sec:empirics}, we extend our framework to capture cross‐impact effects by applying the nonparametric estimation procedure described in Section \ref{sec:estimation}.  Figure \ref{fig:MultiAsset-stocks} displays the estimated cross‐impact kernels for a pair (left panels) and a triplet (right panels) of highly correlated equities.  In each case, the kernels exhibit a pronounced convex decay, and one can observe a cross-impact effect in both directions. One notable strength of the nonparametric estimator in the multi-asset case is that the number of tuned parameters remains zero for for $d$ assets, whereas it is of order $\mathcal{O}(d^2)$ for parametric kernels (see Table \ref{tab:model-parameter-comparison-d-assets}).

We also assess whether the introduction of cross-impact enhances predictive power for nonparametrically estimated impact models. For this, we follow the procedure of \cite{muhle2024stochasticcross} and compare the performance of the following two models:
\begin{itemize}
\item Single-asset projected propagator model (PROJ) as in Section \ref{sec:methods},
where the impact in each asset $\ell\in\{1,\ldots d\}$ only depends on the traded volume in that asset.
\item Multi-asset projected propagator model (CROSS-PROJ), where the impact in each asset $\ell\in\{1,\ldots d\}$ depends not only on the traded volume in that asset, but also on the  traded volume in the overall market. Here, prices, traded volumes, and volatility of the market portfolio are computed as in \cite{muhle2024stochasticcross}, treating it just like an asset. Consequently, the cross-impact propagator matrix is obtained via a nonparametric estimation followed by a projection as before, with the two underlying assets given by the asset $\ell$ and the market portfolio. 
\end{itemize}

Table~\ref{tab:model-comparison-cross} summarizes results for these models. We observe that including cross-impact terms improves the fit for linear models and slightly for square-root models. Additionally, correct concavity choice impacts predictive power significantly, with square-root cross-impact yielding better prediction performance than linear cross-impact. Interestingly, when replacing the market portfolio from Section \ref{sec:methods} with another stock that has a high 10-second return correlation with the underlying stock (\mbox{$\geq 60\%$}), we only observe modest performance gains in terms of out-of-sample $R^2$ on average, indicating that the traded volume in the whole market carries more predictive power than the traded volume in one correlated stock.  

Figure \ref{fig:MultiAsset-stocks} shows the estimated cross-impact kernels for a pair and a triplet of stocks. The left column presents a two-asset model with Coca-Cola (KO) and PepsiCo (PEP), which exhibit a 10-second return correlation of 54\%. Due to similar market capitalizations and liquidity, the cross-impact effects are nearly symmetric in both directions. The self-impact kernels for KO and PEP exhibit a decay pattern consistent with a power-law, reflecting the similar characteristics of these stocks.

The right column illustrates a three-asset model for ConocoPhillips (COP), Chevron (CVX), and Exxon Mobil (XOM), which have pairwise 10-second return correlations of 61\%, 62\%, and 67\%, respectively. Cross-impact and self-impact kernels are calculated highlighting the liquidity distribution. The strongest cross-impact is observed from the more liquid XOM to the less liquid CVX and COP, while impacts in the reverse direction are considerably weaker. This demonstrates the liquidity gradient from XOM to CVX and COP.

\begin{figure}
    \centering
    \includegraphics[width=\linewidth]{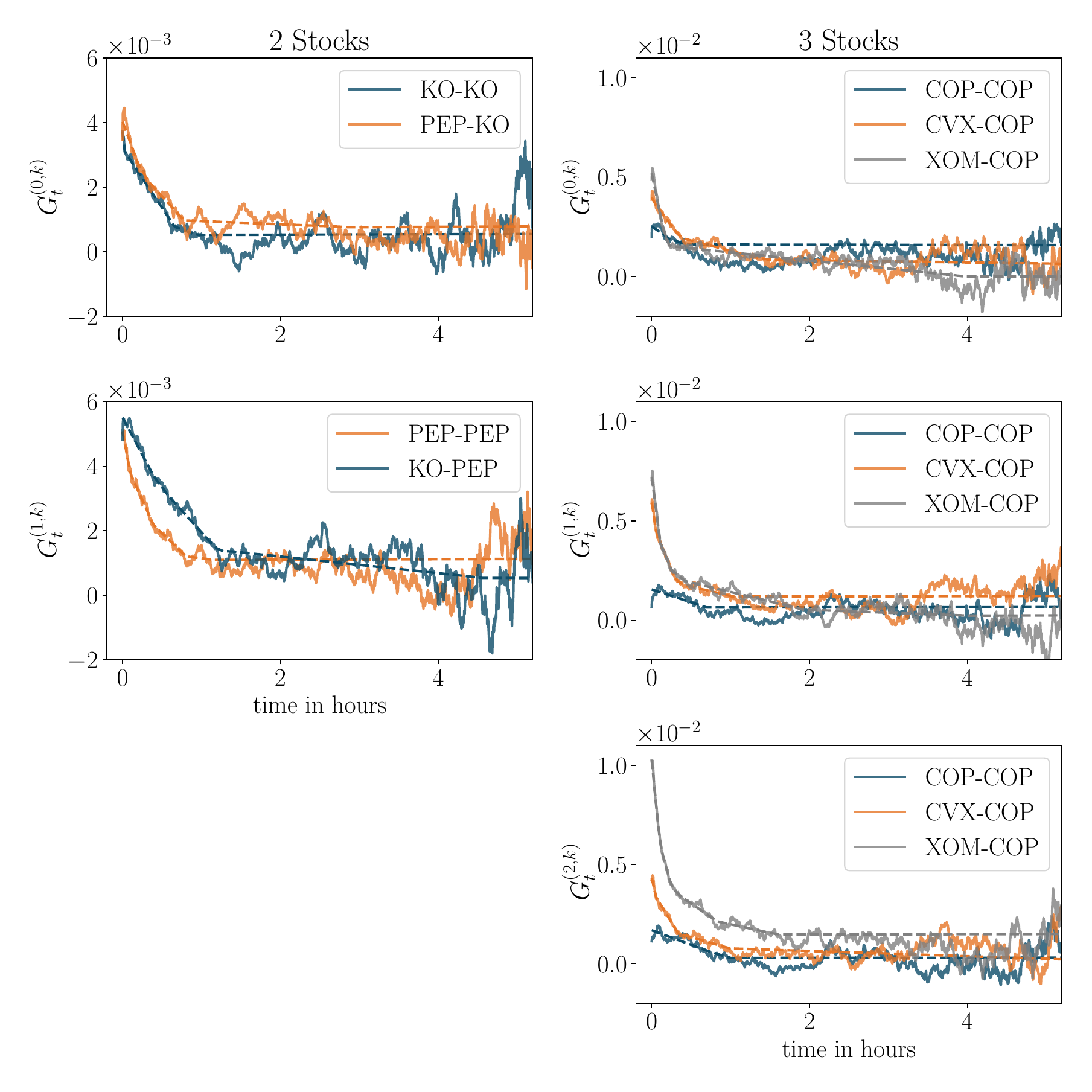}
\caption{Cross-impact kernels are estimated for correlated stocks. The solid and dashed lines depict the raw and projected kernels, respectively. Left column: Results are shown for a two-asset model, where cross-impact is estimated from Coca-Cola (KO) (blue) to itself (above) or to PepsiCo (PEP) (below) and from PEP (orange) to KO (above) or on itself (below). The cross-impact effects are nearly symmetric in both directions and impact decays. Right column: For the three-asset model, cross- and self-impact kernels are similarly estimated across the largest three US oil companies ConocoPhillips (COP), Chevron (CVX), and Exxon Mobil (XOM) (gray). As market capitalization and liquidity decrease from XOM to CVX to COP, the strongest cross-impact is observed flowing from the more liquid XOM into the less liquid CVX and COP, while, conversely, impacts in the opposite direction are markedly weaker.} 
    \label{fig:MultiAsset-stocks}
\end{figure}

\appendix
\section{Robustness Check}
In Section~\ref{sec:empirics-agg}, we aggregated the events using a bin size of 10 seconds. Using the same collection of stocks, we now increase the bin size to 5 minutes to make the aggregate impact analysis comparable to the metaorder impact analysis, 72 data points per trading day after removing the last half-hour. Table \ref{tab:model-comparison-5min}  presents a summary of the in- and out-of-sample performance for different decay kernels using a bin size of 5 minutes.

Compared to the analysis with 10-second bins (see Table \ref{tab:model-comparison}), predictive power declines when using larger bins, likely due to reduced accuracy from data aggregation. However, in comparison to the metaorder impact case (see Section \ref{sec:empirics}), predictive power improves. This enhancement is due to the fact that, for aggregated impact, the signed volume of the entire market serves as a stronger explanatory variable than the signed volume of a metaorder, which represents only a minor perturbation of the market's total order flow.

\begin{table}[htb]
    \centering
    \begin{tabular}{||c | c | c|c | c ||}
    \hline
    Propagator & Concavity $c_S$ &  Horizon & IS $R^2$    & OOS $R^2$ \\
    
    \hline
    \hline
    1-EXP & 1 & 5 & 14.60\% & 9.96\%  \\
    2-EXP & 1 & 5 & 14.95\% & 9.86\% \\
    POWER & 1 & 5 & 14.70\% & 10.01\%\\
    RAW & 1 & 5 & 12.17\% & 7.58\% \\
    PROJ & 1 & 5 & 14.56\%& 9.93\% \\
    1-EXP & 0.5 & 5 & 18.16\% & 17.25\% \\
    2-EXP & 0.5 & 5 & 18.51\% & 17.30\% \\
    POWER & 0.5 & 5 & 18.21\% & 17.35\% \\
    RAW & 0.5 & 5 & 16.56\% & 14.94\% \\
    PROJ & 0.5 & 5 & 18.05\% &17.18\% \\
    \hline
    \end{tabular}
    \caption{In- and out-of-sample performance of the linear and square-root parametric models \mbox{1-EXP}, 2-EXP, POWER and the nonparametric models RAW, PROJ for a bin size of 5 minutes.}
    \label{tab:model-comparison-5min}
     \end{table}

\section{Proofs}\label{appendix:proof}
We now provide the proof of Theorem~\ref{thm:main}, which generalizes the proof of Theorem~2.14 in \cite{neuman2023offline} to the concave multi-asset framework. 
\begin{proof}[Proof of Theorem~\ref{thm:main}]
The proof is similar to the proof of Theorem~2.14 in \cite{neuman2023offline}. Namely, let $(\mathcal{F}_n)_{n \geq 0}$ be the filtration from Definition~\ref{def:dataset}. Then $\boldsymbol{U}^{(n)}$ from \eqref{eq:Un} is $\mathcal{F}_{n-1}$-measurable and $\boldsymbol{y}^{(n)}$ from \eqref{eq:yn} is $\mathcal{F}_{n}$-measurable. Define $X:=\R^{M d^2}$, $Y:=\R^{M d}$, and 
$$
A_n:X\to Y, \quad A_n(v):=\boldsymbol{U}^{(n)}v, \quad n=1,\ldots, N.
$$
Then, recalling \eqref{eq:returns-equation} and Assumption~\ref{assumption}, we get from Theorem~4.5 in \cite{neuman2023offline} that for all $\lambda>0$,
\begin{equation}
\begin{aligned}
&\left\|{W_{N,\lambda}^{1/2} }\big( \boldsymbol{G}_{N,\lambda} - \boldsymbol{G}^* \big) \right\|_{\R^{M d^2}}\\
&\leq R \left( 2 \log \left( \frac{\det(\lambda^{-1}W_{N,\lambda})}{\delta^2}  \right) \right)^{1/2} 
+ \lambda \left\| {W^{-1/2}_{N,\lambda}} \boldsymbol{G}^* \right\|_{\R^{M d^2}}\\
&= R \left( 2 \log \left( \frac{\det(W_{N,\lambda})}{\delta^2\lambda^{M d^2}}  \right) \right)^{1/2} 
+ \lambda \left\| {W^{-1/2}_{N,\lambda}} \boldsymbol{G}^* \right\|_{\R^{M d^2}}.
\end{aligned}
\end{equation}
This proves the desired bound.
\end{proof}
Next, we give the proof of Theorem \ref{thm:price-manipulation}, which extends Proposition 1 of \cite{gatheral.al.12} to the discrete-time framework.
\begin{proof}[Proof of Theorem \ref{thm:price-manipulation}]
Since $h$ is continuous, odd, and not linear, there exist $a,b\in\mathbb{R}$ such that
\be\label{eq:s}
s:=h(a)+h(b)+h(-a-b)\neq 0.
\ee
Otherwise, if $h(a)+h(b)+h(-a-b)=0$ for all $a,b\in\R$, then due to the oddness the identity becomes $h(a+b)=h(a)+h(b)$ for all $a,b\in\R$, which gives a contradiction as the continuity of $h$ then yields $h(x)=qx$ for some $q\in\R$. Moreover, using the oddness again, we can assume $s<0$ without loss of generality.
Fix such $a,b$ and set $x_0:=a$, $x_1:=b$, $x_2=:-a-b$, and $x_3=\ldots=x_{M-1}=0$. Take $t_0=\ldots=t_{M-1}=t^*$. Then
\be\label{eq:S}
S=\sum_{i,j=0}^{M-1} x_i H(t_i,t_j) h(x_j)
=H(t^*,t^*)\Big(\sum_{i=0}^{M=1} x_i\Big)\Big(\sum_{j=0}^{M-1} h(x_j)\Big)=H(t^*,t^*)\cdot 0\cdot s=0,
\ee
with $H(t^*,t^*)>0$ and $s<0$ defined in \eqref{eq:s}. Next, replace $x_0$ by $x_0+\eps$ for some $\eps>0$, so that $\sum_{i=0}^{M-1} x_i=\eps$. By the fact that $s< 0$ and the continuity of $h$ at $x_0$, for sufficiently small $\eps$ we have 
\be\label{eq:sum}
\sum_{j=0}^{M-1} h(x_j)=s+\big(h(x_0+\varepsilon)-h(x_0)\big)<0.
\ee
Hence, \eqref{eq:S} and \eqref{eq:sum} yield $S<0$. Finally, by the continuity of $H$ in $(t^*,t^*)$, choosing the distinct points $t_i:=t^*+i\delta$ for sufficiently small $\delta>0$ preserves negativity of $S$. The second part of the theorem then follows from choosing $H(t,s)=G(|t-s|)$ for all $t,s\in[0,T]$.
\end{proof}
\appendix

\end{document}